\documentclass[11pt,oneside,a4paper]{article}
\usepackage[top=1 in, bottom=1 in, left=0.85 in, right=0.85 in]{geometry}
\usepackage[hidelinks]{hyperref}

\usepackage{multirow}
\usepackage[table,xcdraw]{xcolor}

\usepackage{caption}
\usepackage{subcaption}
\usepackage{algorithm, algorithmic}
\usepackage{amsmath,amssymb,graphicx,url}
\usepackage{thmtools,thm-restate,wrapfig,enumitem,mathabx}
\newtheorem{assumption}{Assumption}
\newtheorem{theorem}{Theorem}
\newtheorem{remark}{Remark}
\newtheorem{lemma}{Lemma}

\newtheorem{proof}{Proof}

\title{Scalable and Sample Efficient Distributed Policy Gradient Algorithms in Multi-Agent Networked Systems}

\author{Xin Liu \\ShanghaiTech University \\ liuxin7@shanghaitech.edu.cn 
\and
Honghao Wei \\ University of Michigan, Ann Arbor \\ honghaow@umich.edu
\and 
Lei Ying \\University of Michigan, Ann Arbor \\ leiying@umich.edu}

\date{}

\begin{document}

\maketitle

\begin{abstract}
This paper studies a class of multi-agent reinforcement learning (MARL) problems where the reward that an agent receives depends on the states of other agents, but the next state only depends on the agent's own current state and action. We name it REC-MARL standing for REward-Coupled Multi-Agent Reinforcement Learning. REC-MARL has a range of important applications such as real-time access control and distributed power control in wireless networks. 
This paper presents a distributed policy gradient algorithm for REC-MARL. The proposed algorithm is {\em distributed} in two aspects: (i) the learned policy is a distributed policy that maps a local state of an agent to its local action and (ii) the learning/training is distributed, during which each agent updates its policy based on its own and neighbors' information. The learned algorithm achieves {\em a stationary policy} and its iterative complexity bounds depend on the dimension of local states and actions. 
The experimental results of our algorithm for the real-time access control and power control in wireless networks show that our policy significantly outperforms the state-of-the-art algorithms and well-known benchmarks.
\end{abstract}

\section{Introduction}
Multi-Agent Reinforcement Learning, or MARL for short, considers a reinforcement learning problem where multiple agents interact with each other and with the environment to finish a common task or achieve a common objective. The key difference between MARL and single-agent RL is that each agent in MARL only observes a subset of the state and receives an individual reward and does not have global information. Examples include multi-access networks where each user senses the collision locally, but needs to coordinate with each other to maximize the network throughput, or power control in wireless networks where each wireless node can only measure its local interference but they need to collectively decide the transmit power levels to maximize a network-wide utility, or task offloading in edge computing where each user observes its local quality of service (QoS) but they coordinate tasks offloaded to edge servers to maintain a high QoS for the network, 
or a team of robots where each robot can only sense its own surrounding environment, but the robot team needs to cooperate and coordinate for accomplishing a rescue task.  

MARL raises two fundamental aspects that are different from single-agent RL. The first aspect is the policy space (the set of policies that are feasible) is restricted to local policies since each agent has to decide an action based on information available to the agent.
The second aspect of MARL is that while distributed learning, e.g. distributed policy gradient, is desired, it is impossible in general MARL because an agent does not know the global state and rewards. One popular approach to address the second issue is the  centralized learning and distributed execution paradigm \cite{FoeFarAfo_18}, where data samples are collected by a central entity and the central entity uses data from all agents to learn a local policy.  Therefore, the learning occurs at the central entity, but the learned policy is a local policy.

In this paper, we consider a special class of MARL, which we call REward-Coupled Multi-Agent Reinforcement Learning (REC-MARL). In REC-MARL, the problem is coupled through the reward functions.  More specifically, the reward function of agent $n$ depends on agent $n$'s state and action and its neighbors' states and actions. However, the next state of agent $n$ only depends on agent $n'$s current state and action, and is independent of other agents' states and actions. In contrast, in general MARL, the transition kernels of the agents are coupled, so the next state of an agent depends on other agents' states and actions. 

While REC-MARL is more restrictive than MARL, a number of applications of MARL are actually REC-MARL. In wireless networks, where agents are network nodes or devices, and the state of an agent is could be its queue length or transmit power, the state only depends on an agent's current state and its current action (see Section \ref{sec: model} for detailed description). 
For REC-MARL, this paper proposes distributed policy gradient algorithms and establish the local convergence (i.e. the algorithms achieve a stationary policy). The main contributions are summarized as follows.
\begin{itemize}[leftmargin=*]
    \item We establish {\em perfect} decomposition of value functions and policy gradients in Lemmas \ref{lem:dec q} and \ref{lem:dec pg}, respectively, where we show that the global value functions (policy-gradients) can be written as a summation of local value functions (policy gradients). This decomposition significantly reduces the complexity of value functions and motivates our distributed multi-agent policy gradient algorithms.   
    
    \item We propose a Temporal-Difference (TD) learning based regularized distributed multi-agent policy gradient algorithm, named TD-RDAC in Algorithm \ref{alg: TD PG}. We proved in Theorem \ref{alg: TD PG} that TD-RDAC achieves the local convergence with the rate $\tilde {O}\left(NS_{\max} A_{\max}/\sqrt{T}\right),$ which only depends on the maximum sizes of local state and action spaces, instead of the sizes of the global action and state spaces. 
    
    \item We apply TD-RDAC to the applications of real-time access and power control in ad hoc wireless networks, both are well-known challenging networking resource management problems. 
    Our experiments show that TD-RDAC significantly outperforms the state-of-the-art algorithm \cite{QuWieLi_22} and well-known benchmarks \cite{TanTamSri_01,Rob_75} in these two applications. 
\end{itemize} 

\subsection*{Related Work}
{\it MARL for networking:} MARL have been applied to various networking applications (e.g.,  content caching \cite{WanWanLiu_20}, packet routing \cite{SunKirRen_21}, video transcoding and transmission \cite{CheXuWan_21}, transportation network control \cite{ChuChiKat_20}). The work \cite{WanWanLiu_20} proposed a multi-agent advantage actor-critic algorithm for  large-scale video caching at network edge and it reduced the content access latency and traffic cost. The proposed algorithm requires the knowledge of neighbors' policies, which is usually not observable. \cite{SunKirRen_21} studied packet routing problem in WAN and applied MARL to minimize the average packet completion time. The method adopted dynamic consensus algorithm to estimate the global reward, which incurs a heavy communication overhead.
The work \cite{CheXuWan_21} proposed a multi-agent actor-critic algorithm for  crowd-sourced livecast services and improved the average throughput and the transmission delay performance. The proposed algorithm requires a centralized controller to maintain the global state information. The most related work is \cite{ChuChiKat_20}, which utilized spatial discount factor to decouple the correlation of distant agents and designed networked MARL solution in traffic signal control and cooperative cruise control. However, it requires a dedicated communication module to maintain a local estimation of global state, which again incurs large communication cost. Moreover, the theoretical performance of these algorithms in \cite{WanWanLiu_20,SunKirRen_21,CheXuWan_21,ChuChiKat_20} are not investigated. 

{\it Provable MARL:} There have been a great amount of works addressing the issues of scalability and sample complexity in MARL.  
A popular paradigm is the centralized learning and distributed execution (see e.g. \cite{FoeFarAfo_18}), where agents share a centralized critic but have individual actors. Both the critic and the actors are trained at a central server, but the actors are local policies that lead to a distributed implementation. \cite{ZhaYanLiu_18} proposed a decentralized actor-critic algorithm for a cooperative scenario, where each agent has its individual critic and actor with a shared global state. The proposed algorithm converges an stationary point of the objective function. Motivated by  \cite{ZhaYanLiu_18}, \cite{CheZhoChe_22} and \cite{HaiLiuLu_22} provided a finite-time analysis of decentralized actor-critic algorithms in the infinity-horizon discounted-reward MDPs and average-reward MDPs, respectively. \cite{ZhaRenLi_21} studied a policy gradient algorithm for multi-agent stochastic games, where each agent maximizes its reward by taking actions independently based on the global state. It established its local convergence for general stochastic games and its global convergence for Markov potential games. The centralized critic or shared states (even with decentralized actors) require collecting global information and centralized training. The work \cite{QuWieLi_20} exploits the network structure to develop a localized or scalable actor-critic (SAC) framework that is proved to achieve $O(\gamma^{k+1})$-approximation of a stationary point in $\gamma$-discounted-reward MDPs. \cite{QuWieLi_22} and \cite{LinQuHua_21} extended SAC framework in \cite{QuWieLi_20} into the settings of infinity-horizon average-reward MDPs and stochastic/time-varying communication graphs. The SAC framework in \cite{QuWieLi_20} is efficient in implementing the paradigm of distributed learning and distributed execution. It is also worth mentioning that in a recent work \cite{ZhaQuXu_22}, the authors studied a kernel-coupled setting and established that the localized policy-gradient converges to a neighborhood of global optimal policy, where the distance to the global optimality depends on the number of hops, $k$, polynomially and could be a constant when $k$ is small.
Another line of research in MARL that addresses the scalability issue is the mean-field approximation (MFA)  \cite{GuoHuXu_19,EliJulLau_20,XieYanWan_21}, where agents are homogeneous and an individual agent interacts or games with the approximated ``mean'' behavior. However, the MFA approach is only applicable to a homogeneous system. Finally, \cite{MeuHauKim_98} and \cite{WeiYuNee_18} studied weakly-coupled MDPs, where individual MDPs are independent and coupled through constraints, instead of coupled reward as in ours. 

Different from these existing works, this paper considers reward-coupled multi-agent
Reinforcement learning (REC-MARL), establishes the local convergence of policy gradient algorithms with both distributed learning and distributed implementation, and demonstrates its efficiency in classical and challenging resource management problems. 

\section{Model}\label{sec: model}
We consider a  multi-agent system where the agents are connected by an interactive graph $\mathcal G = (\mathcal I, \mathcal E)$ with $\mathcal I$ and $\mathcal E$ being the set of nodes and edges, respectively. The system consists of $N = |\mathcal I|$ agents who are connected with edges in $\mathcal E.$ Each agent $n \in\mathcal I$ is associated with a $\gamma$-discounted Markov decision process of $(\mathcal S_n, \mathcal A_n, r_n, \mathcal P_n, \gamma),$ where $\mathcal S_n$ is the state space, $\mathcal A_n$ is the action space, $r_n$ is the reward function, and $\mathcal P_n$ is the transition kernel. Define the neighbourhood of agent $n$ (including itself) to be $\mathcal N(n).$ Define the states and actions of the neighbors of agent $n$ to be $s_{\mathcal N(n)}$ and $a_{\mathcal N(n)},$ respectively. We next formally define the REward-Couple Multi-Agent Reinforcement Learning. 

{\bf REC-MARL:} 
The reward of agent $n$ depends on its neighbors' states and actions $r_n(s_{\mathcal N(n)}, a_{\mathcal N(n)}).$ The transition kernel of agent $n,$ $\mathcal P_n(\cdot |s_n,a_n),$ only depends on its own state $s_n$ and action $a_n.$ Agent $n$ uses a local policy $\pi_{n}: \mathcal S_n \to \mathcal A_n,$ where $\pi_{n}(a_n|s_n)$ is the probability of taking action $a_n$ in state $s_n.$ 

Given a REC-MARL problem, the global state space is  $\mathcal S = \mathcal S_1 \times \mathcal S_2 \times \cdots \times \mathcal S_n$ with $S_{\max} = \max_n |\mathcal S_n|;$ the global action space is $\mathcal A = \mathcal A_1 \times \mathcal A_2 \times \cdots \times \mathcal A_n$ with $A_{\max} = \max_n |\mathcal A_n|;$ the global reward function is $r(s,a) = \frac{1}{N}\sum_{n=1}^N r_n(s_{\mathcal N(n)}, a_{\mathcal N(n)});$ the transition kernel is $\mathcal P(s'|s,a), \forall s,s'\in \mathcal S, \forall a\in \mathcal A$ with $\mathcal P(s'|s,a) = \prod_{n=1}^N \mathcal P_n(s'_n|s_n,a_n).$ 

In this paper, we study the softmax policy parameterized by $\theta_n: \mathcal S_n \times \mathcal A_n \to \mathbb R$ that is $$\pi_{\theta_n}(a_n|s_n) = \frac{e^{\theta_n(a_n, s_n)}}{\sum_{a'_n\in \mathcal{A}_n }{e^{\theta_n(a'_n, s_n)}}}, \forall n.$$ 
The global policy $\pi_{\theta}: \mathcal S \times \mathcal A  \to \mathbb R$ with $\theta = [\theta_1,\theta_2,\cdots, \theta_N]$ is as follows $$\pi_\theta(a|s) = \prod_{n=1}^{N} \pi_{\theta_n}(a_n|s_n).$$
In this paper, we study $\gamma$-discounted infinite-horizon Markov decision processes (MDP) with $(s(t), a(t))$ being the state and action of the MDP at time $t.$ For a global policy $\pi_\theta,$  its value function, action value function ($Q$-function), and advantage function given the initial state and action ($s(0), a(0)$) are defined below: 
\begin{align}
    V^{\pi_\theta}(s) &:= \mathbb E\left[\sum_{t=0}^{\infty}\gamma^t r(s(t),a(t)) \Big| s(0) = s\right], \nonumber\\
    Q^{\pi_\theta}(s,a) &:= \mathbb E\left[\sum_{t=0}^{\infty}\gamma^t r(s(t),a(t)) \Big| s(0) = s, a(0) = a\right], \nonumber \\
    A^{\pi_\theta}(s,a) &:= Q^{\pi_\theta}(s,a) - V^{\pi_\theta}(s). \nonumber
\end{align}
Let $\rho$ be the initial state distribution. Define $V^{\pi_\theta}(\rho) = \mathbb E_{s\sim \rho} [V^{\pi_\theta}(s)].$ Further define the discounted occupancy measure to be
\begin{align}
    d^{\pi_\theta}_\rho(s) = (1-\gamma) \sum_{t=0}^\infty \gamma^t \mathcal P(s(t) = s| \rho, \pi_\theta). \label{eq:discounted ocupy prob}
\end{align}
We seek policies with distributed learning and distributed execution in this paper. We will decompose the global $Q$-function as a sum of local $Q$-functions so that the agents can collectively optimize the local policy.  
Before presenting the details of our solution, we introduce two representative applications in wireless networks: real-time access control and distributed power control. 

{\bf Real-time access control in wireless networks:} 
Consider a wireless access network with $N$ nodes (agents) and $M$ access points as shown in Figure \ref{fig: access}. At the beginning of each time slot, a packet arrives at node $n$ with probability $w_n$. The packet is associated with deadline $d.$ $q_m$ is the successful transmitting probability when transmitting to access point $m.$ A packet is removed if either it is sent to an access point (not necessarily delivered successfully) or it expires. At each time, each node decides whether to transmit a packet to one of the access points or keeping silence. A collision happens if multiple nodes send packets to the same access point simultaneously. Therefore, the throughput of a node depends not only on its decision but also its neighboring nodes' decisions. In particular, the throughput of node $n$ is 
\begin{align}
r_n(a_{\mathcal N(n)})= \sum_{m \in \mathcal {AP}(n)} a_{n, m} \prod_{k\neq n, k\in \mathcal N(n)} (1-a_{k,m}) q_m, \label{eq: rta obj}
\end{align} where $a_{n, m} \in \{0, 1\}$ indicate if  node $n$ transmits a packet to its access point $m \in \mathcal {AC}(n).$ The goal of real-time access control is to maximize the network throughput. 

The problem of real-time access control is  challenging because i) the throughput functions are non-convex and highly-coupled functions of other nodes' decisions and ii) the system parameters are unknown (e.g., $w_n$ and $p_m$). This problem can be formulated as a REC-MARL problem. In particular for node/agent $n,$ the MDP formulation is $(Q_n, a_n, r_n(a_{\mathcal N(n)}), Q'_n):$ the state of agent $n$ is the queue length $Q_n = \{Q_n^l\}_l$ with $Q_n^l$ being the number of packets with the remaining time $l$ ; the action is $a_{n,m}$ with $m\in \mathcal AP(n)$ (one of access point of 
agent $n$); $r_n(a_{\mathcal N(n)})$ is the reward for agent $n;$ $Q'_n$ is its next queue state.

{\bf Distributed power control in wireless network:} The other application of REC-MARL is distributed power control in wireless networks. Consider a network with $N$ wireless links (agents) as shown in Figure \ref{fig: power}. Each link can control the transmission rate by adapting its power level, and neighboring links interfere with each other. Therefore, the transmission rate of a link depends not only on its transmit power but also its neighboring links' transmit power. In particular, 
the reward/utility of link $n$ is its  transmission rate minus a linear power cost: 
\begin{align}
\log\left(1 + \frac{p_nG_{n,n}}{\sum_{m \neq n, m\in \mathcal N(n)} p_m G_{m,n}+\sigma_n}\right) - u_n p_n, \label{eq: pc obj}
\end{align} where $G_{m,n}$ is the channel gain from node $m$ to $n,$ $p_n$ is the power of node $n,$ $\sigma_n$ is the power of noise at user $n,$ and $u_n$ are the trade-off parameters. The goal of distributed power control to maximize the total rewards for the  network.  

The distributed power control is also challenging because the reward function is non-convex and highly-coupled, and the channel condition is unknown and dynamic. 
This problem can also be formulated as a REC-MARL problem. In particular for link/agent $n,$ the MDP formulation is $(p_n, a_n, r_n(p_{\mathcal N(n)}), p'_n):$ the state of agent $n$ is the power level $p_n \in \mathbb Z^{+}$ with $p_n \leq p_{\max}$ ($p_{\max}$ is the maximum power constraint); the action is $a_n\in \{0,-1,+1\}$ (keep, decrease, or increase the power level by one); $r_n(p_{\mathcal N(n)})$ is the reward for agent $n;$ $p'_n$ is the next power level.  

In the following, we present distributed multi-agent reinforcement learning algorithms that can be used to solve these two problems and show that our algorithm outperforms the existing benchmarks. 

\begin{figure}
\centering
\begin{subfigure}{0.35\textwidth}
    \includegraphics[width=\textwidth]{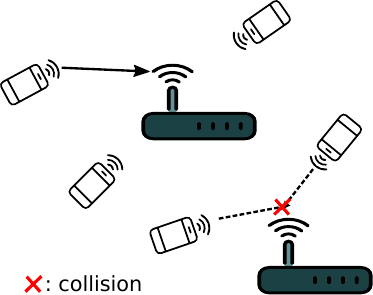}
    \caption{Real-time access control}
    \label{fig: access}
\end{subfigure}
\hspace{50pt}
\begin{subfigure}{0.35\textwidth}
    \includegraphics[width=\textwidth]{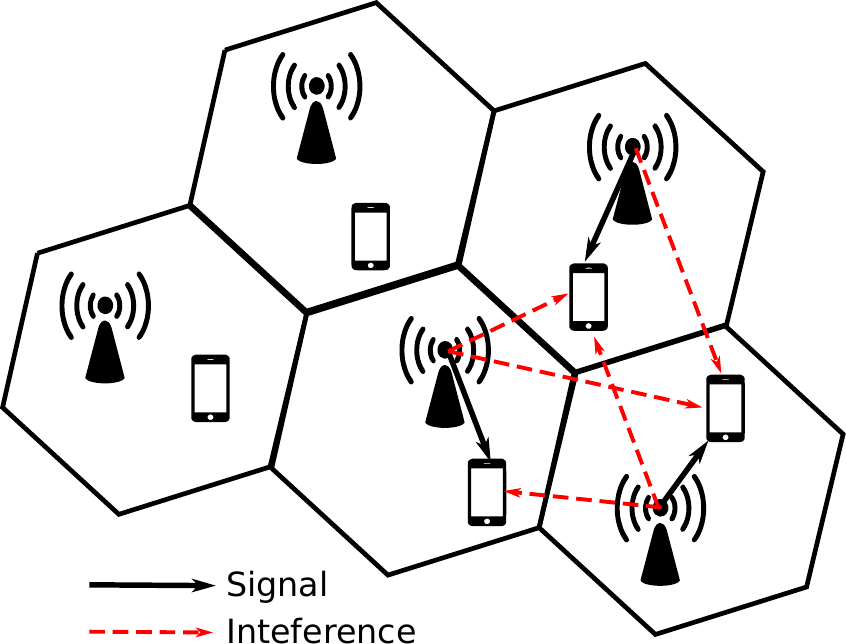}
    \caption{Distributed power control}
    \label{fig: power}
\end{subfigure} \label{fig:traditional}
\caption{Applications of REC-MARL in Wireless Networks}  
\end{figure}

\section{Decomposition of Value and Policy Gradient Functions}
To implement a paradigm of distributed learning and distributed execution, we first establish the decomposition of value and policy gradient functions in Lemma \ref{lem:dec q} and Lemma \ref{lem:dec pg}, respectively. The value and policy gradient functions could be represented locally and computed/estimated via exchange information with their neighbors.  
\subsection{Decomposition of value functions} We decompose the global value ($Q$) functions into the individual value ($Q$) functions and show they only depend on its neighborhood in Lemma \ref{lem:dec q}. The proof could be found in Appendix \ref{app:dec q}.   
\begin{lemma}
Given a multi-agent network, where each agent $n$ is associated with a $\gamma$-discounted Markov decision process defined by $(\mathcal S_n, \mathcal A_n, r_n, \mathcal P_n, \gamma),$ the neighborhood of agent $n$ is defined by $\mathcal{N}(n),$ the network reward function is $r(s,a) = \frac{1}{N}\sum_{n=1}^N r_n(s_{\mathcal N(n)}, a_{\mathcal N(n)})$ and the transition kernel is $\mathcal P(s'|s,a) = \prod_{n=1}^N \mathcal P_n(s'_n|s_n,a_n).$ The policy of a network is $\pi_\theta: \mathcal S \to \mathcal A$ with each agent $n$ using a local policy $\pi_{\theta_n}: \mathcal S_n \to \mathcal A_n,$ we have
\begin{align*}
    V^{\pi_\theta}(s)
    = \frac{1}{N} \sum_{n=1}^N V^{\pi_\theta}_n(s_{\mathcal N(n)}),~~~
    Q^{\pi_\theta}(s,a)
    = \frac{1}{N} \sum_{n=1}^N Q^{\pi_\theta}_n(s_{\mathcal N(n)}, a_{\mathcal N(n)}).
\end{align*}\label{lem:dec q}
\end{lemma}
\begin{remark}
We have decomposed the global value functions exactly into the local value function of individual agent. The decomposition is ``perfect", which is different from the exponential decay property of Lemma $3$ in \cite{QuWieLi_20} or the polynomial decay property of Proposition $2$ in \cite{ZhaQuXu_22}.
\end{remark}

\subsection{Decomposition of Policy Gradient}
We show that policy gradient function can also be decomposed exactly as $Q$-function.
Recall the definitions of $V^{\pi_\theta}(\rho)$ and $d^{\pi_\theta}_\rho,$ and we first state the classical policy gradient theorem in \cite{RicDavSat_99}.
\begin{theorem}[\cite{RicDavSat_99}]
Let $\pi_\theta$ be a policy parameterized by $\theta$ for a $\gamma$-discounted Markov decision process, we have the policy gradient to be
\begin{align*}
\nabla_{\theta} V^{\pi_\theta}(\rho) =& \frac{1}{1-\gamma} \mathbb E_{s\sim d^{\pi_\theta}_\rho, a \sim \pi_\theta(\cdot |s)} \left[ Q^{\pi_\theta}(s,a) \nabla_{\theta} \log \pi_\theta(a|s)\right]. 
\end{align*}\label{thm:pg}
\end{theorem}
Theorem \ref{thm:pg} has motivated policy gradient methods in the single-agent setting (e.g., \cite{KonTsi_20}) and multi-agent setting (e.g., \cite{LowWuTam_17}). However, it is infeasible to be applied for a large-scale multi-agent network due to the large global state space. Therefore, we decompose the policy gradient in the following lemma. The proof can be found in Appendix \ref{app:dec pg}. 
\begin{lemma}
Given a multi-agent network, where each agent $n$ is associated with a $\gamma$-discounted Markov decision process defined by $(\mathcal S_n, \mathcal A_n, r_n, \mathcal P_n, \gamma),$ the neighborhood of agent $n$ is defined by $\mathcal{N}(n),$ the network reward function is $r(s,a) = \frac{1}{N}\sum_{n=1}^N r_n(s_{\mathcal N(n)}, a_{\mathcal N(n)})$ and the transition kernel is $\mathcal P(s'|s,a) = \prod_{n=1}^N \mathcal P_n(s'_n|s_n,a_n).$ The policy of a network is $\pi_\theta: \mathcal S \to \mathcal A$ with each agent $n$ using a local policy $\pi_{\theta_n}: \mathcal S_n \to \mathcal A_n,$ we have
\begin{align*}
   \nabla_{\theta_n} V^{\pi_\theta}(\rho)
   =& \frac{1}{1-\gamma} \mathbb E_{s\sim d^{\pi_\theta}_\rho, a \sim \pi_\theta(\cdot |s)} \left[ \frac{1}{N} \sum_{k\in \mathcal N(n)} Q^{\pi_\theta}_k(s_{\mathcal N(k)}, a_{\mathcal N(k)}) \nabla_{\theta_n} \log \pi_{\theta_n}(a_n|s_n)\right] \\
   =& \frac{1}{1-\gamma} \mathbb E_{s\sim d^{\pi_\theta}_\rho, a \sim \pi_\theta(\cdot |s)} \left[ \frac{1}{N} \sum_{k\in \mathcal N(n)} A^{\pi_\theta}_k(s_{\mathcal N(k)}, a_{\mathcal N(k)}) \nabla_{\theta_n} \log \pi_{\theta_n}(a_n|s_n)\right] 
\end{align*}\label{lem:dec pg}
\end{lemma}
\begin{remark}
Lemma \ref{lem:dec pg} implies that policy gradient of agent $n$ could be computed/estimated via exchanging information $Q^{\pi_\theta}(s_{\mathcal N(k)},a_{\mathcal N(k)})$ or $A^{\pi_\theta}(s_{\mathcal N(k)},a_{\mathcal N(k)}), \forall k\in \mathcal N(n)$ with its neighbors. This decomposition is the key to implement the efficient paradigm of distributed learning and distributed execution and motivate the algorithms in the following sections.   
\end{remark}

\section{Distributed Multi-Agent Policy Gradient Algorithms}
In this section, we introduce distributed policy gradient based algorithms (DPG) for the multi-agent system. 
Before introducing the algorithms, we assume that the initial state distribution $\mu$ induced by  our policy is strictly positive:
\begin{assumption}
The initial state distribution $\mu$ satisfies $\min_{s_n} \mu(s_n) > 0.$ for any agent $n.$
\end{assumption}
This assumption imposes the sufficient exploration
for the state space for each agent and is common in the literature \cite{AgaKakLee_20,MeiXiaSze_20}. 
We first propose a regularized distributed policy gradient algorithm assuming the access of an inexact gradient, and then introduce TD-learning methods to estimate the gradient.

\subsection{Regularized Distributed Policy Gradient Algorithm with Inexact Gradient}\label{sec:DMAPG inexact}
Assume an estimated gradient $\nabla \hat V^{\pi_{\theta^t}}(\rho)$ at any time $t,$ we study relative entropy regularized objective for the network as in \cite{AgaKakLee_20,ZhaKinODo_20} that  $$L_\lambda(\theta) = V^{\pi_\theta}(\rho) + \sum_{n=1}^N\frac{\lambda}{|\mathcal S_n||\mathcal A_n|}\sum_{s_n,a_n}\log \pi_{\theta_n}(a_n|s_n),$$ where $\lambda$ is a positive constant and $\log \pi_{\theta_n}(a_n|s_n)$ is the regularizer to prevent the probability of taking action $a_n$ at state $s_n$ approaches to a small quantity for any agent $n.$ 
With the regularized value function $L_\lambda(\theta),$ we present a DPG with the inexact gradient in Algorithm \ref{alg: inexact PG}, where the inexact gradient is usually estimated with TD-learning methods (see Section \ref{sec:DMAAC}).
\begin{algorithm}
    \caption{Distributed Multi-Agent Policy Gradient with Inexact Gradient}
	\begin{algorithmic}[1]
		\STATE {\bf Initialization:}   parameters $\eta,$ $\lambda,$ and $\theta^0_n, \forall n \in \mathcal I.$
		\FOR{$t = 1,\dots,T$}
		\STATE{\bf Estimate Policy Gradient:} $\nabla_{\theta_n} \hat L_\lambda(\theta^t)$ for each agent $n \in \mathcal I.$
		\STATE{\bf Update:} $\theta_n^{t+1} = \theta_n^{t} + \eta \nabla_{\theta_n} \hat L_\lambda(\theta^t).$
		\ENDFOR
	\end{algorithmic}
	\label{alg: inexact PG}
\end{algorithm}

By analyzing the dynamics of $\nabla L^{\pi_{\theta^t}}(\rho)$ in Algorithm \ref{alg: inexact PG}, we can establish the upper bound on the cumulative $\mathbb E\left[||\nabla  L^{\pi_{\theta^t}}(\rho)||^2\right]$ in Theorem \ref{thm: inexact local}. The proof can be found in Appendix \ref{app:inexact local}.
\begin{theorem}\label{thm: inexact local}
Let $\beta' = \frac{48N^2}{(1-\gamma)^3} + \sum_{n}\frac{2\lambda}{|\mathcal S_n|}$ and $\epsilon_t = \mathbb E\left[|| \nabla \hat L^{\pi_{\theta^t}}(\rho) - \nabla L^{\pi_{\theta^t}}(\rho)||^2\right].$ We have Algorithm \ref{alg: inexact PG} with the learning rate $\eta \leq 1/\beta'$ such that
$$\sum_{t=1}^T \mathbb E \left[ ||\nabla L^{\pi_{\theta^t}}(\rho)||^2\right] \leq \sum_{t=1}^T \epsilon_t + 2\beta'\left(L^*(\rho) - L^{\pi^1}(\rho)\right).$$
\end{theorem}
Theorem \ref{thm: inexact local} implies a local convergence result of Algorithm \ref{alg: inexact PG} given a reasonable good gradient estimator, e.g., $\sum_{t=1}^T\epsilon_t = O(\sqrt{T}),$ which is related to the quality of estimated value functions according to Lemma \ref{lem:dec pg}. 
Next, we utilize TD-learning to estimate value functions, provide an actor-critic type of algorithm, and establish its local convergence. 

\subsection{Regularized Distributed Actor Critic Algorithm}\label{sec:DMAAC}
Motivated by \cite{QuWieLi_20}, we utilize TD-learning to estimate the gradient $\nabla_{\theta_n} \hat V^{\pi_\theta}(\rho)$  according to Lemma \ref{lem:dec pg} and provide an actor-critic type of algorithm (TD-RDAC) for multi-agent setting in Algorithm \ref{alg: TD PG}.
From Lemma \ref{lem:dec pg}, we have 
\begin{align*}
&\nabla_{\theta_n} L^{\pi_\theta}(\rho) 
= \nabla_{\theta_n} V^{\pi_\theta}(\rho) + \frac{\lambda}{|\mathcal S_n||\mathcal A_n|}\sum_{s_n,a_n}\nabla_{\theta_n}\log \pi_{\theta_n}(a_n|s_n)\\
&=
\frac{1}{1-\gamma} \mathbb E_{s\sim d^{\pi_\theta}_\rho} \left[ \frac{1}{N} \sum_{k\in \mathcal N(n)} A^{\pi_\theta}(s_{\mathcal N(k)}, a_{\mathcal N(k)}) \nabla_{\theta_n} \pi_{\theta_n}(a_n|s_n)\right] + \frac{\lambda}{|\mathcal S_n|}\left(\frac{1}{|\mathcal A_n|} - \pi_{\theta_n}\right).
\end{align*}
For each agent $n,$ it requires to aggregate the advantage functions (or the estimator) from its neighbors, i.e., $A^{\pi_\theta}(s_{\mathcal N(k)}, a_{\mathcal N(k)})$ for any $k \in N(n)$ to estimate the gradient $\nabla_{\theta_n} \hat V^{\pi_\theta}(\rho)$ or $\nabla_{\theta_n} \hat L^{\pi_\theta}(\rho).$ The actor-critic algorithm is summarized in Algorithm \ref{alg: TD PG}. It contains $T$ outer loops and $H$ inner loops. The lines $4$ to $8$ perform TD learning to estimate value function for each individual agent. The line $10$ is to compute TD-error based on the learned value function for each agent. The line $11$ is to estimate the policy gradient by aggregating its neighbors' TD-error. The line $12$ is to update the policy parameters. The inner loops (lines $3$ to $13$) will be run $T$ rounds. 
\begin{algorithm}[H]
    \caption{Regularized Distributed Actor-Critic Algorithm}
	\begin{algorithmic}[1]
		\STATE {\bf Initialization:} $\theta^0,\lambda,$ and $\eta.$
		\FOR {$t = 1,\dots,T$}
		\STATE {\bf Initialize $\hat V_0^{\pi_{\theta^t}}(\cdot) = 0$ and sample a state $s^0 \sim \rho.$} 
		\FOR {$h = 1,\dots,H$}
		\FOR {$n = 1,\dots,N$}
		\STATE {\bf Update value function for each agent at step $h:$}
		\begin{align}
        &\hat V^{\pi_{\theta^t},h+1}_{n}(s_{\mathcal N(n)}(h)) \nonumber \\
        &= (1-\alpha)\hat  V^{\pi_{\theta^t},h}_{n}(s_{\mathcal N(n)}(h)) + \alpha \left(r_n^h(s_{\mathcal N(n)}(h), a_{\mathcal N(n)}(h)) + \gamma \hat V^{\pi_{\theta^t},h}_{n}(s_{\mathcal N(n)}(h+1)\right), \nonumber 
        \end{align}
		\ENDFOR
		\ENDFOR
		\FOR{$n = 1,\dots,N$}
		\STATE {\bf Compute TD-error at each step $h$:
		\begin{align}
        &\delta^{\pi_{\theta^t}}_n(s_{\mathcal N(n)}(h), a_{\mathcal N(n)}(h)) \nonumber \\
        &= r_n^h(s_{\mathcal N(n)}(h), a_{\mathcal N(n)}(h)) + \gamma \hat V^{\pi_{\theta^t},H}_{n}(s_{\mathcal N(n)}(h+1)) - \hat V^{\pi_{\theta^t},H}_{n}(s_{\mathcal N(n)}(h)), \nonumber 
        \end{align}
		}
		\STATE {\bf Estimate policy gradient by aggregating its neighbors' TD-errors:
		\begin{align}
        &\nabla_{\theta_n} \hat V^{\pi_{\theta^t}}(\rho) \nonumber \\
        &= \frac{1}{(1-\gamma)} \sum_{h=1}^H\gamma^h \frac{1}{N}\sum_{k\in \mathcal N(n)} \hat \delta^{\pi_{\theta^t}}_k(s_{\mathcal N(k)}(h), a_{\mathcal N(k)}(h)) \nabla_{\theta_n} \log \pi_{\theta_n}(a_n(h)|s_n(h)), \nonumber 
        \end{align}}
		\STATE {\bf Update:} $\theta_n^{t+1} = \theta_n^{t} + \eta \nabla_{\theta_n} \hat V^{\pi_{\theta^t}}(\rho)+ \frac{\lambda}{|\mathcal S_n|}\left(\frac{1}{|\mathcal A_n|} - \pi_{\theta_n}\right).$
		\ENDFOR
	\ENDFOR
	\end{algorithmic}
	\label{alg: TD PG}
\end{algorithm}
Intuitively, Algorithm \ref{alg: TD PG} requires a large number of inner loops $H$ to guarantee the convergence of value function such that the estimated policy gradient is accurate. 
Before proving the local convergence results of Algorithm \ref{alg: TD PG}, we introduce a common Assumption \ref{assump: fast mixing} (e.g., \cite{SriYin_19,QuLinAda_20}) on the mixing time of Markov decision process. 
\begin{assumption}\label{assump: fast mixing}
Assume $\{s_n(h)\}_h, \forall n,$ be a Markov chain with the geometric mixing time under any policy $\pi_{\theta_n},$ i.e., there exists a constant $K$ such that for any $h \geq K\log (1/\epsilon),$ the total variance distance of $\mathcal P(z^h_n | z^0_n)$ and $\pi_{\theta_n}(z^h_n)$ satisfies
$$\|\mathcal P(z^h_n | z^0_n) - \pi(z^h_n)\|_{TV} \leq \epsilon, \forall n.$$
\end{assumption}
To avoid the complicated conditions on $T$ and $H,$ we present the order-wise results in Theorem \ref{thm: main TD}. The proofs and the detailed requirements on $T$ and $H$ can be found in Appendix \ref{app:TDMAPG}. 
\begin{theorem}\label{thm: main TD}
Suppose large $T$ and $H=O(T).$ Let the learning rate be $\eta = O(1/\sqrt{T}).$ We have Algorithm \ref{alg: TD PG} such that
\begin{align*}
\min_{1\leq t \leq T} \mathbb E[ ||\nabla L^{\pi_{\theta^t}}(\rho)||^2 ] 
\leq& O\left(\frac{NS_{\max}A_{\max}}{(1-\gamma)^4} \sqrt{\frac{\log T}{T}}\right)
\end{align*}
\end{theorem}
Theorem \ref{thm: main TD} concludes 
Algorithm \ref{alg: TD PG} will return a stationary policy for a large $T.$ Moreover, the bound only depends on the local dimension of states and actions, i.e., $S_{\max}$ and $A_{\max}$, which means our algorithm is scalable and could be applied to large-scale networks. 

\section{Experiments}\label{sec:sim}
In this section, we evaluate TD-RDAC for real-time access control problem and distributed power control problem in wireless networks as described in Section \ref{sec: model}. 
\begin{figure}
\centering
\includegraphics[scale=0.73]{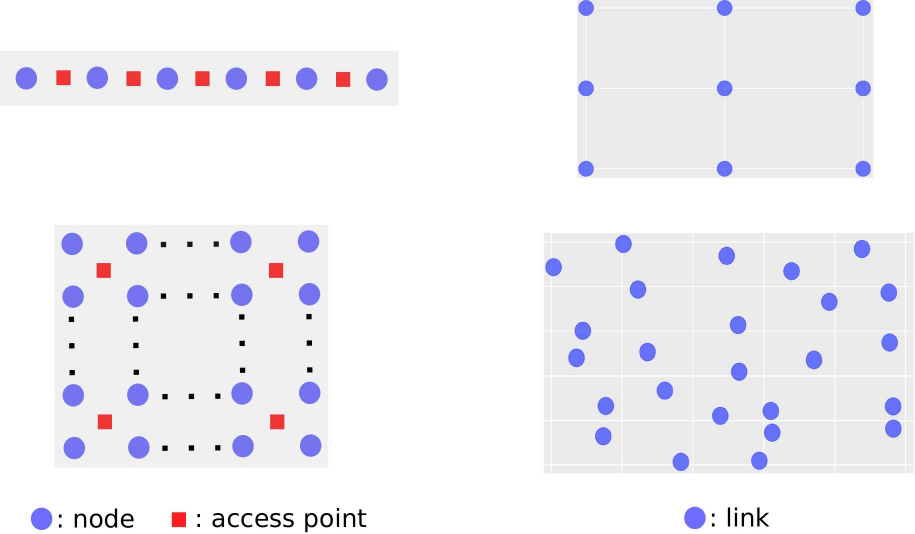}
\caption{Wireless network topology for real-time access control (left-hand side) and power control (right-hand side).}
\label{fig: topology}
\end{figure}

{\noindent \bf Real-time access control in wireless network:} We consider the reward/utility of node $n$ defined in \eqref{eq: rta obj}.  
We compared Algorithm \ref{alg: TD PG} with the ALOHA algorithm \cite{Rob_75} and a scalable actor critic (SAC) in \cite{QuWieLi_22}.
The classical distributed solution to this problem is the localized ALOHA algorithm \cite{Rob_75}, where node $n$ sends the earliest packet with a certain probability and the packet is sent to a random access point $m \in \mathcal {AP}(n)$ in its available set, with probability proportional to $q_m$ and inversely proportional to the number of nodes that can potentially transmit to access point $m.$ Note to compare ALOHA algorithm \cite{Rob_75} and SAC in \cite{QuWieLi_22}, we apply Algorithm \ref{alg: TD PG} into a slightly different setting where a packet is removed from its queue if either it is delivered successfully or expired. 

In our experiment, we considered two types of network topology: a line network and a grid network as shown in the lefthhand side of Figure \ref{fig: topology}. The tabular method is used to represent value functions since the size of the table is tractable in this experiment. 

The line network has $6$ nodes and $5$ access points as shown in the left-up of Figure \ref{fig: topology}. We considered two settings in the line network with the same arrival probability but distinct success transmission probabilities to represent reliable/unreliable environments. Specifically, the arrival probabilities are $w = [0.5, 0.3, 0.5, 0.5, 0.3, 0.5];$ the transmission probabilities of the reliable/unreliable environments are $q = [0.9, 0.95, 0.9, 0.95, 0.9]$ and $q = [0.5, 0.6, 0.7, 0.6, 0.5],$ respectively. The deadline of packets is $d=2.$  We ran $9$ different instances with different random seeds and the shaded region represents the $95\%$ confidence interval. We observe our algorithm increases steady in Figures \ref{fig: line high} and \ref{fig: line low} and outperforms both SAC algorithm in \cite{QuWieLi_22} and ALOHA algorithm in the reliable and unreliable environment, where the converged rewards of (our algorithm v.s. SAC v.s. ALOHA) are ($0.814$ v.s. $0.735$ v.s. $0.334$) and ($0.503$ v.s. $0.464$ v.s. $0.222$) for the reliable and unreliable environment, respectively.  

The grid network is similar as in \cite{QuWieLi_22}, shown in the left-bottom of Figure \ref{fig: topology}. The network has $36$ nodes and $25$ access points. The arrival probability $w_n$ of node $n$ and success transmission probability $q_m$ of access point $m$ are generated uniformly random from $[0, 1].$ The deadline of packets is also set $d=2.$
We observe our algorithm again increases steady in Figure \ref{fig: 36 nodes} even with the relatively dense environment and outperforms both SAC algorithm in \cite{QuWieLi_22} and ALOHA algorithm \cite{Rob_75}, where the converged rewards of (our algorithm v.s. SAC v.s. ALOHA) are ($0.393$ v.s. $0.369$ v.s. $0.138$). 

\begin{figure}[H]
\centering
\begin{subfigure}{0.32\textwidth}
    \includegraphics[width=\textwidth]{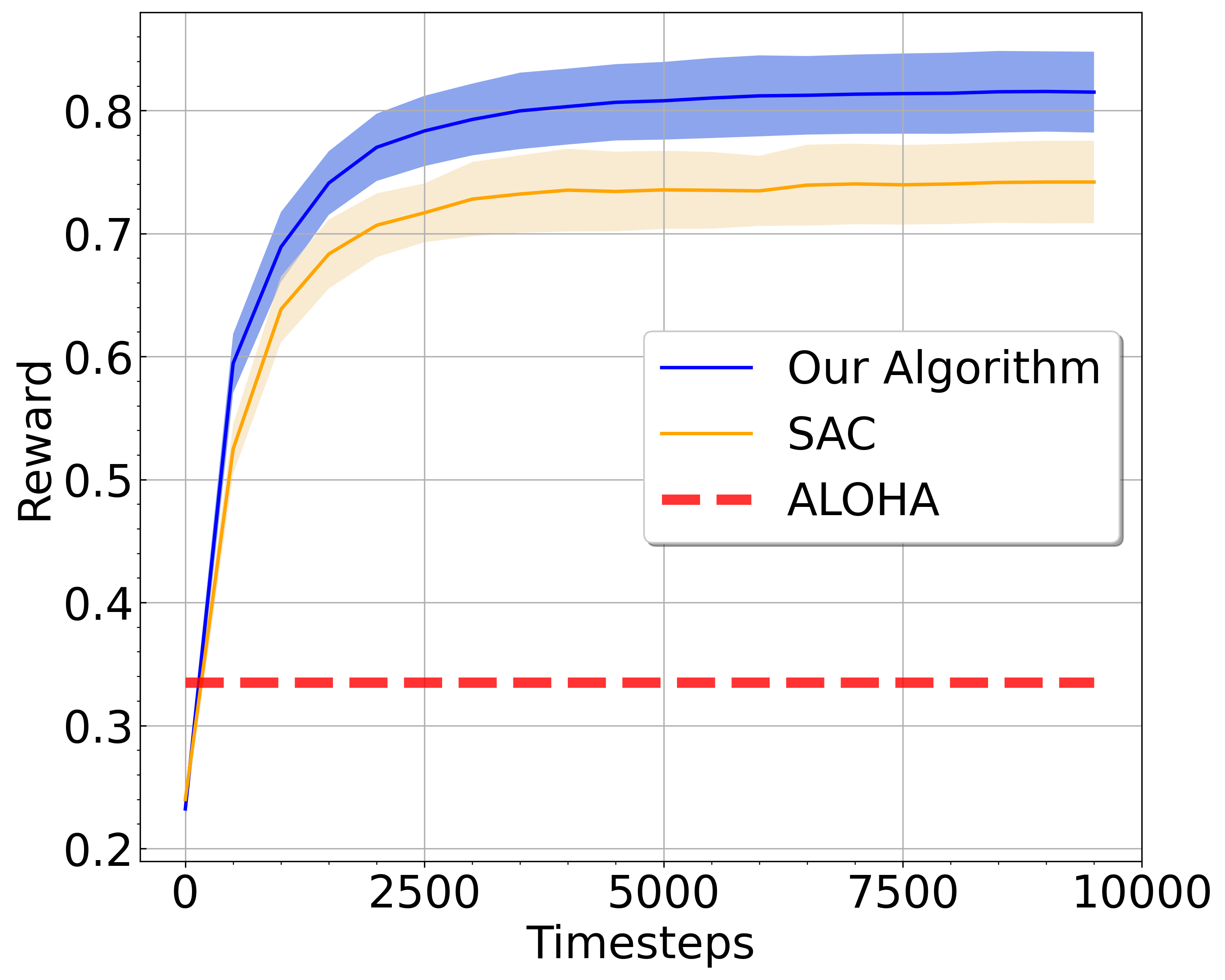}
    \caption{reliable line access network}
\label{fig: line high}
\end{subfigure}
\hfill
\begin{subfigure}{0.32\textwidth}
    \includegraphics[width=\textwidth]{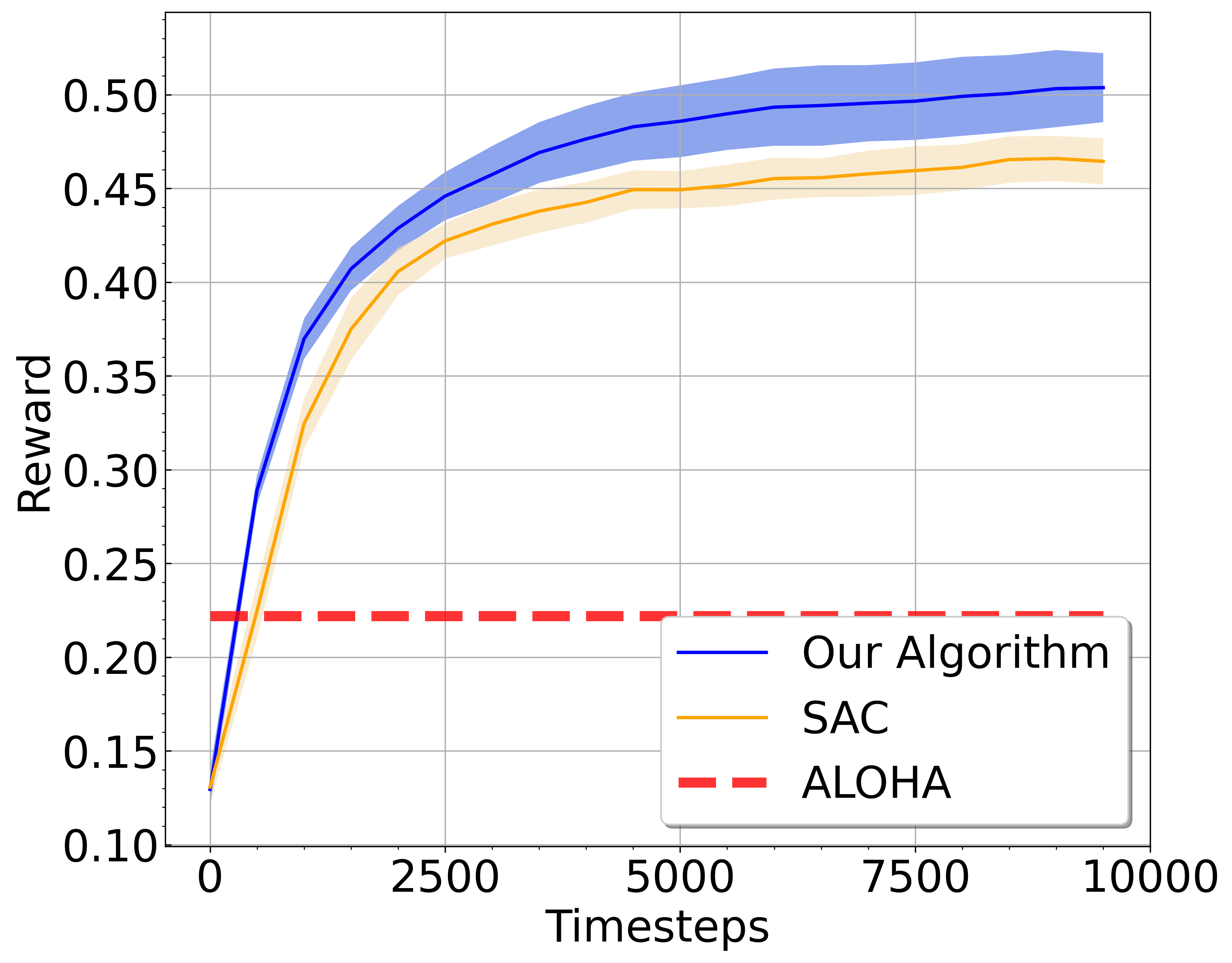}
    \caption{unreliable line access network}
\label{fig: line low}
\end{subfigure}
\hfill
\begin{subfigure}{0.32\textwidth}
    \includegraphics[width=\textwidth]{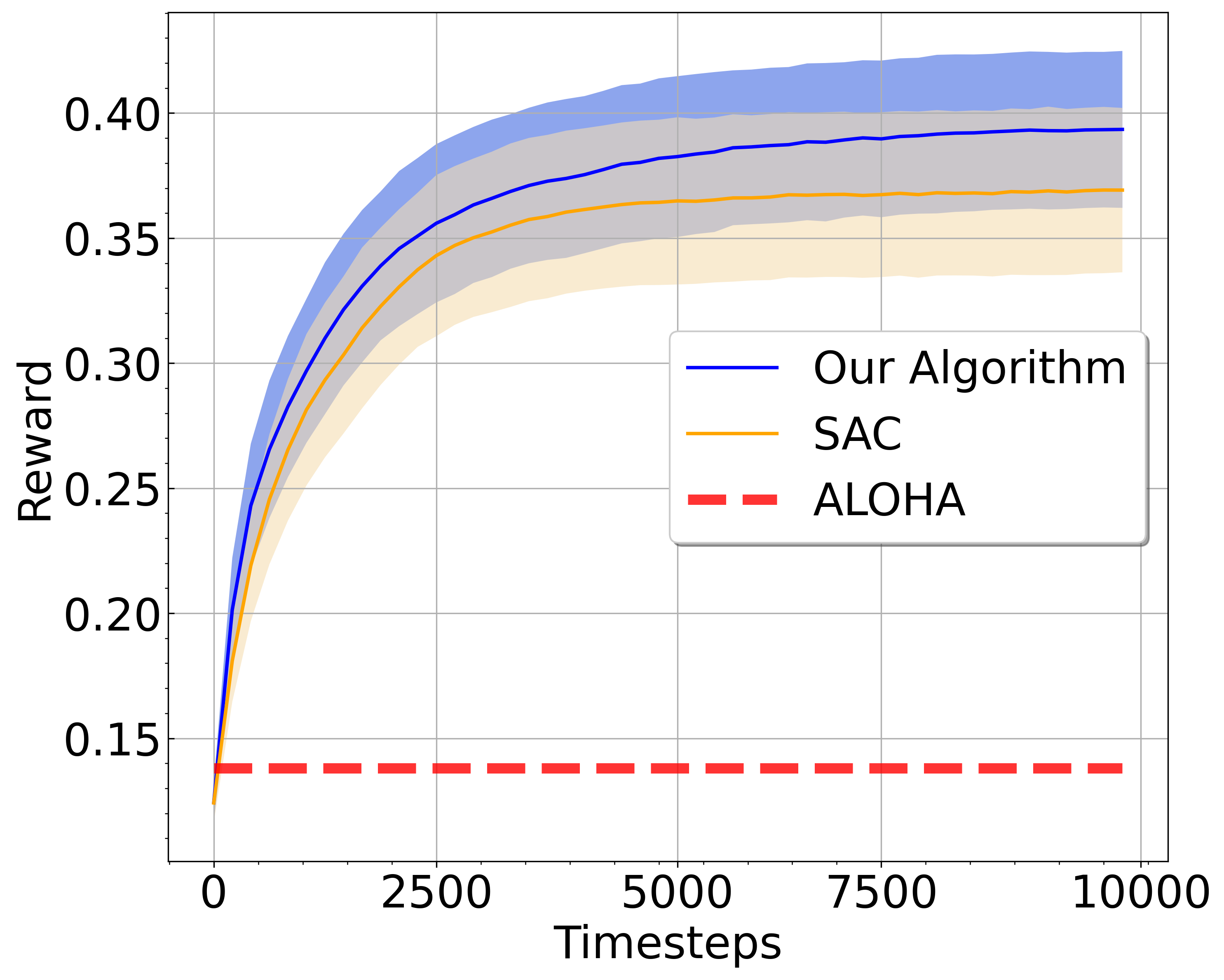}
    \caption{grid access network}
\label{fig: 36 nodes}
\end{subfigure} 
\caption{Performance Comparison in Real-time Access Network}
\end{figure}

\begin{figure}[H]
\centering
\begin{subfigure}{0.32\textwidth}
    \includegraphics[width=\textwidth]{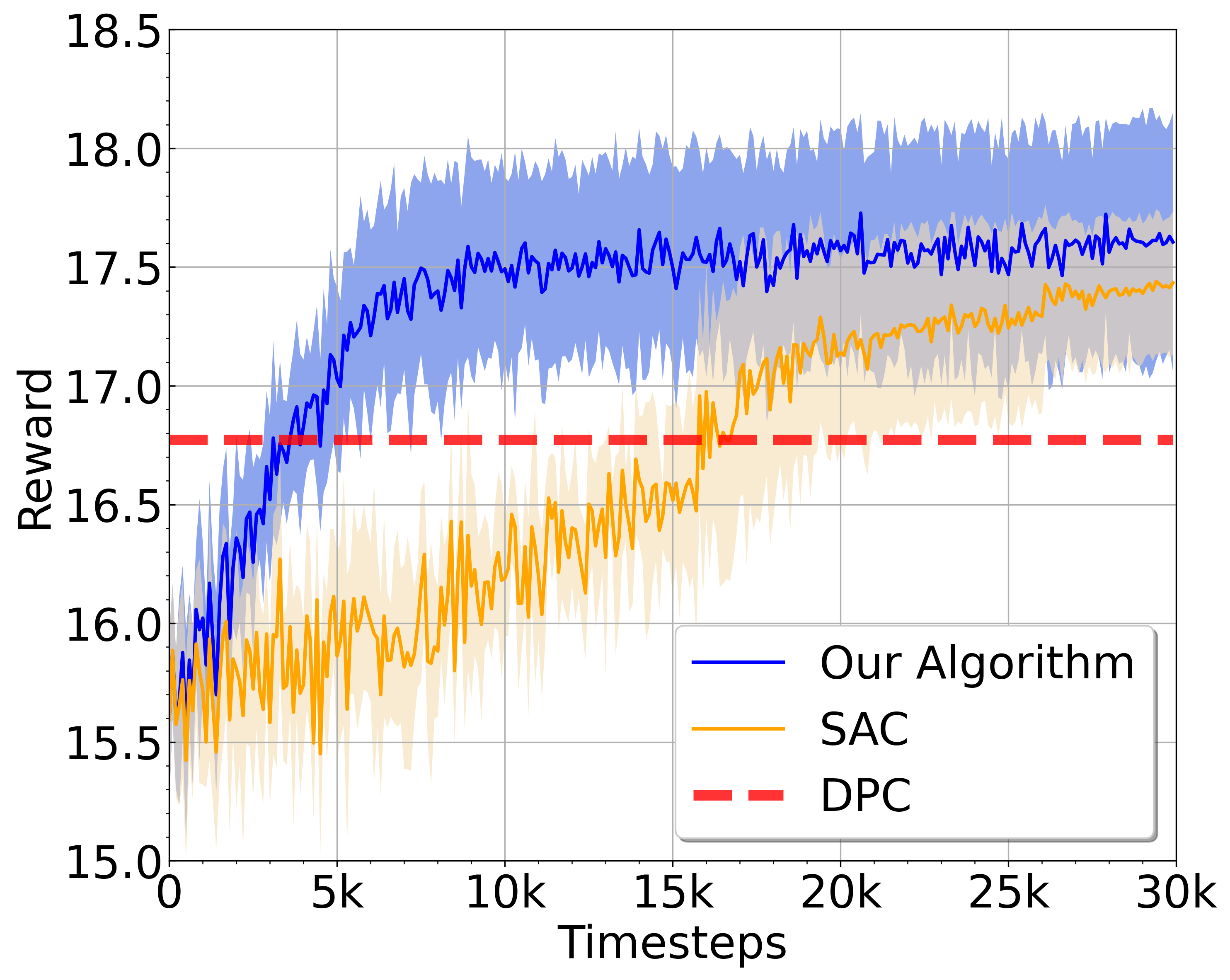}
    \caption{6-link network}
\label{fig: perf 6 nodes}
\end{subfigure}
\hfill
\begin{subfigure}{0.32\textwidth}
    \includegraphics[width=\textwidth]{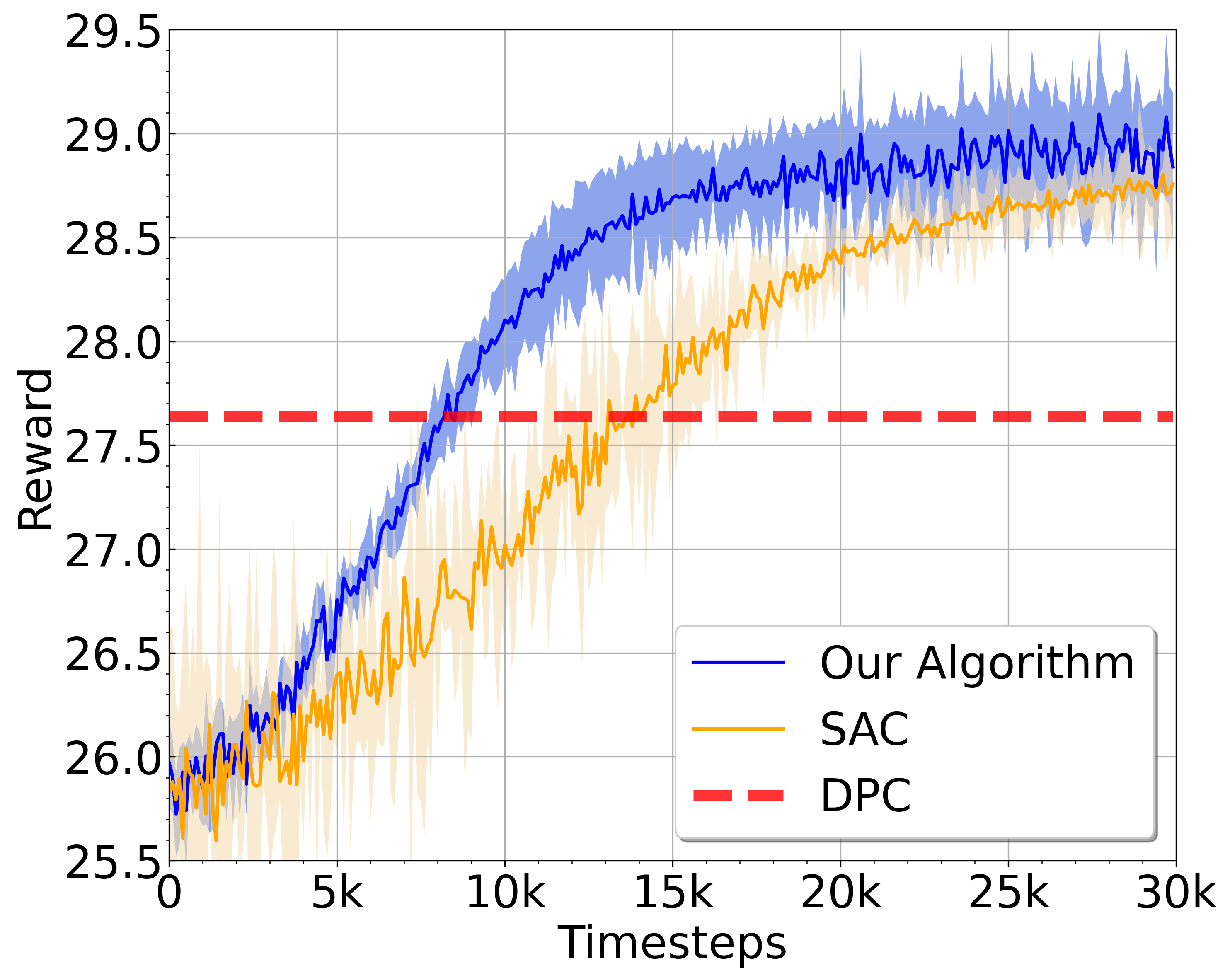}
    \caption{9-link network}
\label{fig: perf 9 nodes}
\end{subfigure}
\hfill
\begin{subfigure}{0.32\textwidth}
    \includegraphics[width=\textwidth]{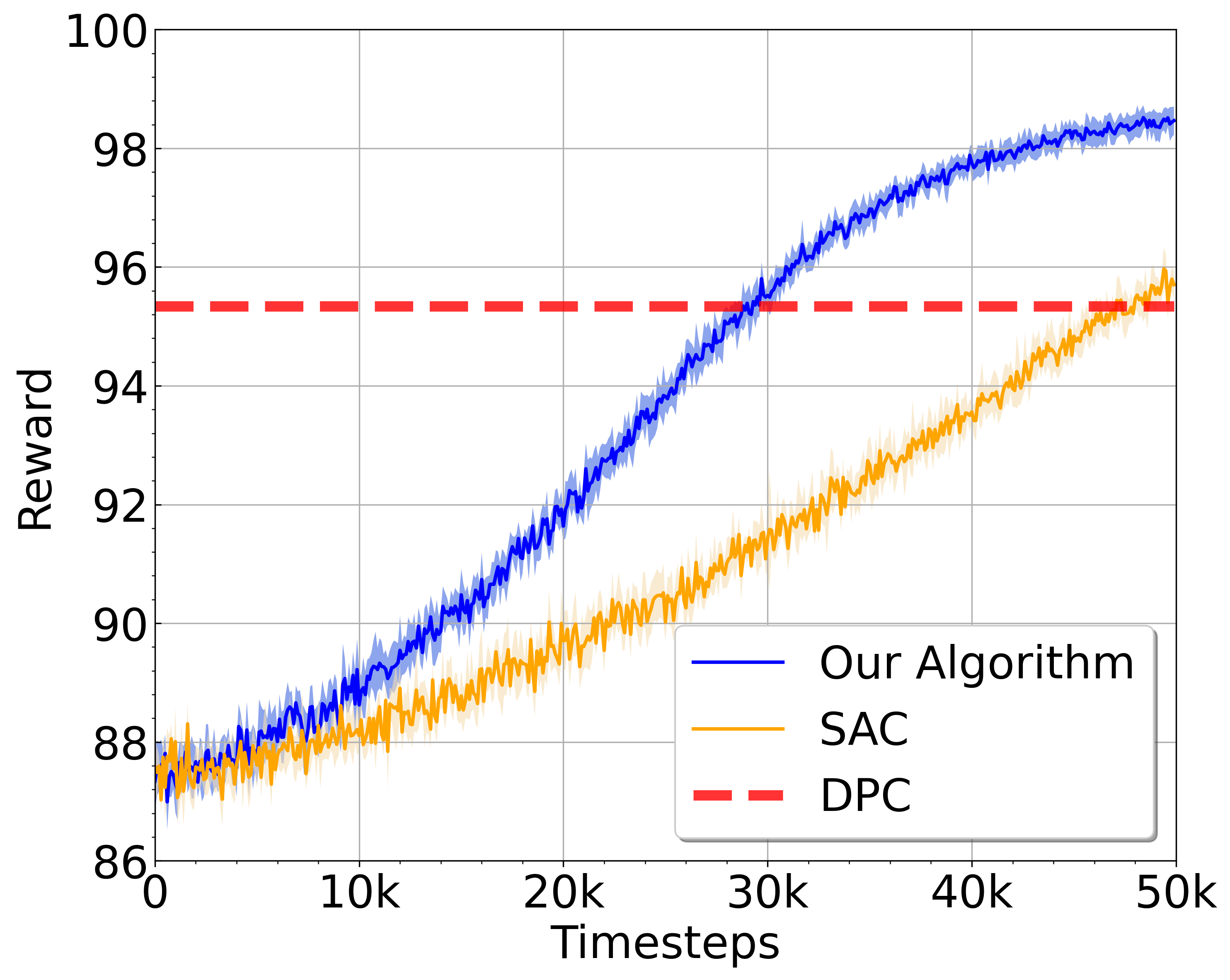}
    \caption{25-link network}
\label{fig: perf 25 nodes}
\end{subfigure}
\caption{Performance Comparison in Distributed Power Control}
\end{figure}

{\noindent \bf Distributed power control in wireless network:} We considered the reward/utility of each agent is the difference between a logarithmic function of signal-interference-noise ratio (SINR) and a linear pricing function proportional to the transmitted power as in \eqref{eq: pc obj}. 
In the literature, the problem has been formulated as a linear pricing game \cite{TanTamSri_01} and \cite{ChiHanLan_08}. 
We compared Algorithm \ref{alg: TD PG} with the DPC algorithm \cite{TanTamSri_01} and SAC in \cite{QuWieLi_22}.
In our experiment, we first studied a setting with $6$-link ($3 \times 2$ grids) 
where the tabular method is tractable to represent value functions. Then we studied two settings: $9$-link 
and $25$-link as shown in the right-hand side of Figure \ref{fig: topology}, where we use blue circles to denote links. For both settings, we use neural network (NN) approximation for value functions. With NN approximation, we utilized the replay buffer \cite{MniKavSil_13} and double-Q learning \cite{HadArtDav_16} to stabilize the learning process. 

For the reward/utility in \eqref{eq: pc obj}, we set $G_{n,n} = 1$ and  $G_{m,n} = \kappa / ({dist}_{m,n})^2$ with $\kappa=0.1$ as the fading factor and ${dist}_{m,n}$ is the distance between nodes $m$ and $n.$ The average power of the noise is $\sigma_n = 0.1.$ We set the trade-off parameter $u_n=0.1, \forall n.$ For each environment, we ran $9$ different instances with different random seeds and the shaded region represents the $ 95\%$ confidence interval.

For $6$-nodes network, 
we represent value functions with tables in Algorithm \ref{alg: TD PG}. We observe the learning processes are quite stable in Figure \ref{fig: perf 6 nodes}. We observe our algorithm outperforms the SAC and DPC algorithm as the training time increases: the converged reward are ($17.60$ v.s. $17.43$ v.s. $16.77$). 
For 9-nodes
, we approximate value functions with neural networks in Algorithm \ref{alg: TD PG}.  We observe Algorithm \ref{alg: TD PG} outperforms the SAC and DPC algorithms as the training time increases in Figure \ref{fig: perf 9 nodes} ($28.92$ v.s. $28.76$ v.s. $27.64$).
For 25 nodes, Figure \ref{fig: perf 25 nodes} shows the results with a heterogeneous topology. 
Again, we observe our algorithm increases steady in Figure \ref{fig: perf 25 nodes} and outperforms the SAC and DPC algorithms ($98.35$ v.s. $95.70$ v.s. $95.34$). 

Finally, we note that both ALOHA and DPC algorithms are model-based algorithm and the agents need to know the systems parameters (success transmitting probability in the real-time access control problem and channel gains and interference levels in the power control problem) while our algorithm is model-free and  does not need to know these parameters. 

\section{Conclusion}
In this paper, we studied a weakly-coupled multi-agent reinforcement learning problem where agents’ reward functions are coupled but the transition kernels are not. We propose TD-learning based regularized multi-agent policy actor-critic algorithm (TD-RDAC) that achieve the local convergence result. We demonstrated the effectiveness of TD-RDAC with two important applications in wireless networks: real-time access control and distributed power control and showed that our
algorithm outperforms the existing benchmarks in the literature.

\subsection*{Acknowledgements} 
The authors are very grateful to Rui Hu for his insightful discussion and comments. 

\bibliographystyle{plain}
\bibliography{ref.bib}

\newpage
\appendix

\section{Decomposition of Value Functions}
\subsection{Proof of Lemma \ref{lem:dec q}}\label{app:dec q}
\begin{proof}
We provide the proof of the decomposition on $Q$-function and the similar argument holds for the value function $V^{\pi_\theta}(s).$
According to the definitions of reward and $Q$-function, we have
\begin{align*} 
    Q^{\pi_\theta}(s,a) =& \frac{1}{N} \sum_{n=1}^N \mathbb E\left[\sum_{t=1}^{\infty}\gamma^t r_n(s_{\mathcal N(n)}^t, a_{\mathcal N(n)}^t) \Big| s^0 = s, a^0 = a\right] \\
    =& \frac{1}{N} \sum_{n=1}^N \mathbb E\left[\sum_{t=1}^{\infty}\gamma^t r_n(s_{\mathcal N(n)}^t, a_{\mathcal N(n)}^t) \Big|  s_{\mathcal N(n)}^0 = s_{\mathcal N(n)},  a_{\mathcal N(n)}^0 = a_{\mathcal N(n)}\right] \\
    =& \frac{1}{N} \sum_{n=1}^N Q^{\pi_\theta}_n( s_{\mathcal N(n)}, a_{\mathcal N(n)})
\end{align*}
where the second equality holds because the reward and transition of agent $n$ only depends on its neighbors' states and actions since for any trajectories of agent $n$ until $t',$ we have 
\begin{align*}
    \mathcal P(s^{0}, a^{0},\cdots, s_{\mathcal N(n)}^{t'}, a_{\mathcal N(n)}^{t'}) 
    &= \prod_{t=1}^{t'} \prod_{k\in \mathcal N(n)} {\pi_{\theta_k}(a_k^t|s_k^t) \mathcal P( s_k^t | s_k^{t-1},a_k^{t-1})} \\
    &= \prod_{k\in \mathcal N(n)} \prod_{t=1}^{t'}  {\pi_{\theta_k}(a_k^t|s_k^t) \mathcal P( s_k^t | s_k^{t-1},a_k^{t-1})}\\
    &= \mathcal P(s_{\mathcal N(n)}^{0}, a_{\mathcal N(n)}^{0},\cdots, s_{\mathcal N(n)}^{t'}, a_{\mathcal N(n)}^{t'}),
\end{align*}
and it implies $\mathcal P(s_{\mathcal N(n)}^{t'}, a_{\mathcal N(n)}^{t'})$ only depends on $(s_{\mathcal N(n)}^{0}, a_{\mathcal N(n)}^{0}).$ 
\end{proof}
\subsection{Proof of Lemma \ref{lem:dec pg}}\label{app:dec pg} 
\begin{proof}
According to Theorem \ref{thm:pg}, we have
\begin{align*}
   \nabla_{\theta_n} V^{\pi_\theta}(\rho) =& \frac{1}{1-\gamma} \mathbb E_{s\sim d^{\pi_\theta}_\rho, a \sim \pi_\theta(\cdot |s)} \left[ Q^{\pi_\theta}(s,a) \nabla_{\theta_n} \log \pi_\theta(a|s)\right] \\
   =& \frac{1}{1-\gamma} \mathbb E_{s\sim d^{\pi_\theta}_\rho, a \sim \pi_\theta(\cdot |s)} \left[ Q^{\pi_\theta}(s,a) \nabla_{\theta_n} \log \pi_{\theta_n}(a_n|s_n)\right] \\
   =& \frac{1}{1-\gamma} \mathbb E_{s\sim d^{\pi_\theta}_\rho, a \sim \pi_\theta(\cdot |s)} \left[ \frac{1}{N} \sum_{n=1}^N Q^{\pi_\theta}(s_{\mathcal N(n)},a_{\mathcal N(n)}) \nabla_{\theta_n} \log \pi_{\theta_n}(a_n|s_n)\right] \\
   =& \frac{1}{1-\gamma} \mathbb E_{s\sim d^{\pi_\theta}_\rho, a \sim \pi_\theta(\cdot |s)} \left[ \frac{1}{N} \sum_{k\in \mathcal N(n)} Q^{\pi_\theta}(s_{\mathcal N(k)},a_{\mathcal N(k)}) \nabla_{\theta_n} \log \pi_{\theta_n}(a_n|s_n)\right] 
\end{align*}
where the second equality holds because of localized policies $\pi_{\theta}(a|s) = \prod_n \pi_{\theta_n}(a_n|s_n);$ the third equality holds because of Lemma \ref{lem:dec q}; the last equality holds because for any $k \notin \mathcal N(n)$ that
\begin{align*}
   &\mathbb E_{s\sim d^{\pi_\theta}_\rho, a \sim \pi_\theta(\cdot |s)} \left[ Q^{\pi_\theta}(s_{\mathcal N(k)},a_{\mathcal N(k)}) \nabla_{\theta_n} \log \pi_{\theta_n}(a_n|s_n)\right] \\
   =& \mathbb E_{s\sim d^{\pi_\theta}_\rho} \left[  \sum_{a}Q^{\pi_\theta}(s_{\mathcal N(k)},a_{\mathcal N(k)}) \prod_{j\neq n} \pi_{\theta_j}(a_j|s_j)\nabla_{\theta_n}  \pi_{\theta_n}(a_n|s_n)\right]\\
   =& \mathbb E_{s\sim d^{\pi_\theta}_\rho} \left[  \sum_{a_{-n}}Q^{\pi_\theta}(s_{\mathcal N(k)},a_{\mathcal N(k)}) \prod_{j\neq n} \pi_{\theta_j}(a_j|s_j) \sum_{a_n}\nabla_{\theta_n}  \pi_{\theta_n}(a_n|s_n)\right]\\
   =& 0 
\end{align*}
Note $\nabla_{\theta_n} V^{\pi_\theta}(\rho) = \frac{1}{1-\gamma} \mathbb E_{s\sim d^{\pi_\theta}_\rho, a \sim \pi_\theta(\cdot |s)} \left[ A^{\pi_\theta}(s,a) \nabla_{\theta_n} \log \pi_\theta(a|s)\right]$ and the second equality on the advantage function holds by following the same steps. 
\end{proof}

\section{Proof of Theorem \ref{thm: inexact local}}\label{app:inexact local}
To prove Theorem \ref{thm: inexact local}, we first prove the smoothness of regularized value $L$ function.
\subsection{Smoothness of $L$ function}\label{app:inexact smooth}
Recall the definition of $L$ function $$L_\lambda(\theta) = V^{\pi_\theta}(\rho) + R(\theta).$$
where $R(\theta) =\sum_{n=1}^N\frac{\lambda}{|\mathcal S_n||\mathcal A_n|}\sum_{s_n,a_n}\log \pi_{\theta_n}(a_n|s_n).$
We show $V^{\pi_\theta}(\rho)$ is $\frac{48N^2}{(1-\gamma)^3}$-smoothness in Lemma \ref{lem: smooth} and $R(\theta)$ is $\sum_{n}\frac{2\lambda}{|\mathcal S_n|}$-smoothness in Lemma \ref{lem: smooth regular}, respectively, which imply $L^{\pi_\theta}(s)$ is $\beta' :=\frac{48N^2}{(1-\gamma)^3}+ \sum_{n}\frac{2\lambda}{|\mathcal S_n|}$ smoothness.  
\begin{lemma}[Smoothness]\label{lem: smooth}
Under Algorithm \ref{alg: inexact PG}, we have $V^{\pi_\theta}(s)$ to be $48N^2/(1-\gamma)^3$ smoothness, i.e.,
$$\left|V^{\pi_\theta}(s) - V^{\pi_{\theta'}}(s) - \langle \frac{\partial V^{\pi_\theta}(s)}{\partial \theta}, \theta-\theta' \rangle \right| \leq \frac{24N^2}{(1-\gamma)^3}||\theta-\theta'||^2.$$
\end{lemma}

\begin{proof}
The Bellman equation (for a network) is 
$$V^{\pi_\theta}(s) = \sum_{a} \pi_\theta(a|s)r(s,a) + \gamma \sum_a \pi_\theta(a|s) \sum_{s'} \mathcal P(s'|s,a) V^{\pi_\theta}(s'),$$
Let $$r_\theta(s) = \sum_{a} \pi(a|s)r(s,a) ~~\text{and}~~ \mathcal P_{\theta}(s,s') = \sum_{a}\pi_\theta(a|s) \sum_{s'} \mathcal P(s'|s,a).$$
Let $u_s$ be a vector with only the $s$th entry being one and zeros for other entries. The Bellman equation can be wrote as  
\begin{align*}
    V^{\pi_\theta}(s) 
    =& r_\theta(s) + \gamma \langle \mathcal P_\theta(s, \cdot),V^{\pi_\theta}(\cdot)\rangle,
\end{align*}
which implies $V^{\pi_\theta} \in \mathbb R^{|\mathcal S|}$ satisfies $$V^{\pi_\theta} = (I - \gamma \mathcal P(\theta))^{-1}r_{\theta} := M_\theta r_{\theta},$$
where $I$ is the identity matrix.
Let $\theta_\alpha = \theta + \alpha \nu,$ and we have $\nu^\dagger\frac{\partial^2 V^{\pi_\theta}(s)}{\partial \theta^2} \nu = \frac{\partial^2 V^{\pi_{\theta_\alpha}}(s) }{\partial \alpha^2} |_{\alpha=0}.$ With a bit abuse of notation, let $r_{\alpha} = r_{\theta_\alpha},$ $\mathcal P_{\alpha}=\mathcal P_{\theta_\alpha},$ $\pi_{\alpha} = \pi_{\theta_\alpha},$ and $M_{\alpha}=M_{\theta_\alpha}.$ We need to compute the first-order gradient
\begin{align*}
    \frac{\partial V^{\pi_{\theta_\alpha}}(s)}{\partial \alpha} =& \gamma u_s^T\frac{\partial  M_\alpha }{\partial \alpha} r_{\alpha} + u_s^T M_\alpha \frac{\partial r_{\alpha}}{\partial \alpha} \\
    =& \gamma u_s^TM_\alpha\frac{\partial \mathcal P_\alpha }{\partial \alpha}M_\alpha r_{\alpha} + u_s^T M_\alpha \frac{\partial r_{\alpha}}{\partial \alpha}
\end{align*}
where we use the derivative of inverse matrix $\frac{\partial (I - \gamma \mathcal P(\theta))^{-1}}{\partial \alpha} = (I - \gamma \mathcal P(\theta))^{-1}\frac{\partial (I - \gamma \mathcal P(\theta))}{\partial \alpha}(I - \gamma \mathcal P(\theta))^{-1};$ and the second-order gradient
\begin{align*}
    \frac{\partial^2 V^{\pi_{\theta_\alpha}}(s)}{\partial \alpha^2} 
    =& 2\gamma^2 u_s^TM_\alpha\frac{\partial \mathcal P_\alpha }{\partial \alpha}M_\alpha \frac{\partial \mathcal P_\alpha }{\partial \alpha}M_\alpha r_{\alpha} + \gamma u_s^TM_\alpha\frac{\partial^2  \mathcal P_\alpha }{\partial \alpha^2}M_\alpha r_{\alpha} \\
    &+ \gamma u_s^T M_\alpha \frac{\partial \mathcal P_\alpha}{\partial \alpha} M_\alpha\frac{\partial r_{\alpha}}{\partial \alpha} + u_s^T M_\alpha \frac{\partial^2 r_{\alpha}}{\partial \alpha^2}.
\end{align*}
Therefore, we need to establish the upper bounds on $$\frac{\partial \mathcal P_\alpha }{\partial \alpha},~~\frac{\partial^2 \mathcal P_\alpha }{\partial \alpha^2},~~\frac{\partial  r_\alpha }{\partial \alpha},~~\frac{\partial^2  r_\alpha }{\partial \alpha^2},$$
which require to compute $$\frac{\partial  \pi_\alpha }{\partial \alpha},~~\frac{\partial^2  \pi_\alpha }{\partial \alpha^2}.$$
Recall $\pi_\theta(a|s) = \prod_{n=1}^{N} \pi_{\theta_n}(a_n|s_n).$ Let's compute 
\begin{align*}
    \frac{\partial  \pi_\alpha(a|s) }{\partial \alpha} |_{\alpha=0} &= \langle \frac{\partial  \pi_\theta(a|s) }{\partial \theta(s,\cdot)}, \nu(s,\cdot)\rangle = \sum_{n=1}^N \langle \frac{\partial  \pi_\theta(a|s) }{\partial \theta_n(s_n,\cdot)}, \nu_n(s_n,\cdot)\rangle \\
    &= \sum_{n=1}^N \prod_{k\neq n} \pi_{\theta_k}(a_k|s_k)~ \langle \frac{\partial  \pi_{\theta_n}(a_n|s_n) }{\partial \theta_n(s_n,\cdot)}, \nu_n(s_n,\cdot)\rangle \\
    &=\sum_{n=1}^N \pi_\theta(a|s) \left(\nu_n(s_n, a_n) -  \langle \pi_{\theta_n}(\cdot|s_n), \nu_n(s,\cdot)\rangle\right)
\end{align*}
where the last equality holds because 
$$\frac{\partial  \pi_{\theta_n}(a_n|s_n) }{\partial \theta_n(s_n,a'_n)} = \pi_{\theta_n}(a_n|s_n)(\mathbb I_{a_n=a'_n} - \pi_{\theta_n}(a'_n|s_n)),$$
and it implies
$$ \langle \frac{\partial  \pi_{\theta_n}(a_n|s_n) }{\partial \theta_n(s_n,\cdot)}, \nu_n(s_n,\cdot)\rangle = \nu_n(s_n, a_n) -  \langle \pi_{\theta_n}(\cdot|s_n), \nu_n(s,\cdot)\rangle.$$
Therefore, we have 
\begin{align*}
    \sum_a \left|\frac{\partial  \pi_\alpha(a|s) }{\partial \alpha} |_{\alpha=0} \right|
    &= \sum_a \left|\sum_{n=1}^N \pi_\theta(a|s) \left(\nu_n(s_n, a_n) -  \langle \pi_{\theta_n}(\cdot|s_n), \nu_n(s_n,\cdot)\rangle\right)\right| \\
    &\leq 2 \max_a \sum_{n=1}^N |\nu_n(s_n,a_n)| \\
    &\leq 2N ||u||_2.
\end{align*}
Next, let's compute 
\begin{align}
    \frac{\partial^2  \pi_\alpha(a|s) }{\partial \alpha^2} |_{\alpha=0} &= \sum_{n=1}^N \sum_{n'=1}^N \langle \frac{\partial^2  \pi_\theta(a|s) }{\partial \theta_n(s_n,\cdot)\theta_{n'}(s_{n'},\cdot)}\nu_n(s_n,\cdot), \nu_{n'}(s_{n'},\cdot)\rangle, \nonumber
\end{align}
where we need to compute  
\begin{align*}
    \left(\frac{\partial^2  \pi_\theta(a|s) }{\partial \theta_n(s_n,\cdot)^2} \right)_{i,j} =& \mathbb I_{a_n=i} \pi_\theta(a|s)\left(\mathbb I_{a_n=j} - \pi_{\theta_n}(j|s_n)\right)
    -\pi_\theta(a|s) \pi_{\theta_n}(i|s_n)\left(\mathbb I_{i=j} - \pi_{\theta_n}(j|s_n)\right)\\
    & -\pi_\theta(a|s) \pi_{\theta_n}(i|s_n)\left(\mathbb I_{a_n=j} - \pi_{\theta_n}(j|s_n)\right),
\end{align*}
and 
\begin{align*}
    \left(\frac{\partial^2  \pi_\theta(a|s) }{\partial \theta_n(s_n,\cdot)\theta_{n'}(s_{n'},\cdot)} \right)_{i,j} =& \pi_\theta(a|s)\left(\mathbb I_{a_n=i} - \pi_{\theta_n}(i|s_n)\right) \left(\mathbb I_{a_{n'}=j} - \pi_{\theta_{n'}}(j|s_{n'})\right).
\end{align*}
Then we have
\begin{align*}
&\left|\langle \frac{\partial^2  \pi_\theta(a|s) }{\partial \theta_n(s_n,\cdot)^2}\nu_n(s_n,\cdot), \nu_n(s_n,\cdot)\rangle / \pi(a|s) \right| \\
\leq & \nu^2_n(s_n,a_n) + 2|\nu_n(s_n,a_n)\sum_{j} \pi_{\theta_n}(j|s_n)\nu_n(s_n,j)| \\
&+ |\sum_{i}\pi_{\theta_n}(i|s)\nu^2_n(i|s)
+ 2\sum_i\sum_j \pi_{\theta_n}(i|s_n)\pi_{\theta_n}(j|s_n) \nu_n(s_n,i)\nu_{n'}(s_{n'},j)|\\ 
\leq&  6 \max_{a_n} \nu^2_n(s_n,a_n), 
\end{align*}
and
\begin{align}
& \left|\langle \frac{\partial^2  \pi_\theta(a|s) }{\partial \theta_n(s_n,\cdot)\theta_{n'}(s_{n'},\cdot)}\nu_n(s_n,\cdot), \nu_{n'}(s_{n'},\cdot)\rangle/\pi(a|s)\right|\nonumber\\
\leq &  |\nu_n(s_n,a_n)\nu_{n'}(s_{n'},a_{n'})| + |\nu_n(s_n,a_n)\sum_{j} \pi_{\theta_{n'}}(j|s_{n'})\nu_{n'}(s_{n'},j)| \\
&+ |\nu_{n'}(s_{n'},a_{n'})\sum_{i} \pi_{\theta_{n}}(i|s_{n})\nu_{n}(s_{n},i)|
+ \sum_i\sum_j |\pi_{\theta_n}(i|s_n)\pi_{\theta_{n'}}(j|s_{n'}) \nu_n(s_n,i)\nu_{n'}(s_{n'},j)|\\ 
\leq&  4 \max_{s_n,s_{n'},a_n,a_{n'}} |\nu_n(s_n,a_n)\nu_{n'}(s_{n'},a_{n'})|
\end{align}
which implies 
\begin{align}
    \left|\frac{\partial^2  \pi_\alpha(a|s) }{\partial \alpha^2} |_{\alpha=0}\right| \leq 6N^2 ||\nu||_2^2  
\end{align}
Now it is good to bound 
\begin{align*}
    \left|\frac{\partial^2 V^{\pi_{\theta_\alpha}}(s)}{\partial \alpha^2}|_{\alpha=0}\right| 
    =& 2\gamma^2\left|u_s^TM_\alpha\frac{\partial \mathcal P_\alpha }{\partial \alpha}M_\alpha \frac{\partial  \mathcal P_\alpha }{\partial \alpha}M_\alpha r_{\alpha}\right| + \gamma \left|u_s^TM_\alpha\frac{\partial^2  \mathcal P_\alpha }{\partial \alpha^2}M_\alpha r_{\alpha}\right|\\ 
    &+  \gamma\left|u_s^T M_\alpha \frac{\partial \mathcal P_\alpha }{\partial \alpha} M_\alpha\frac{\partial r_{\alpha}}{\partial \alpha}\right| + \left|u_s^T M_\alpha \frac{\partial^2 r_{\alpha}}{\partial \alpha^2}\right|\\
    \leq& \left[\frac{8N^2}{(1-\gamma)^3} + \frac{6N^2}{(1-\gamma)^2} + \frac{4N^2}{(1-\gamma)^2} + \frac{6N^2}{(1-\gamma)^2} \right] ||\nu||_2^2\\
    \leq& \frac{24N^2}{(1-\gamma)^3}||\nu||_2^2.
\end{align*}
Finally, we have $$\nu^\dagger\frac{\partial^2 V^{\pi_\theta}(s)}{\partial \theta^2} \nu \leq \frac{24N^2}{(1-\gamma)^3}||\nu||_2^2,$$ which completes the proof.
\end{proof}

\begin{lemma}\label{lem: smooth regular}
Under Algorithm \ref{alg: inexact PG}, we have $R(\theta)$ is $\sum_{n}\frac{2\lambda}{|\mathcal S_n|}$-smoothness.
\end{lemma}
\begin{proof}
We have for any $n$ such that
\begin{align}
    \frac{\partial R(\theta)}{\partial \theta_n(s_n,\cdot)} =& \frac{1}{|\mathcal S_n|}\left(\frac{1}{|\mathcal A_n|} - \pi_{\theta_n}(\cdot|s_n)\right), \nonumber \\
    \frac{\partial^2 R(\theta)}{\partial \theta_n(s_n,\cdot)^2} =& \frac{1}{|\mathcal S_n|}\left(-\text{diag}(\pi_{\theta_n}(\cdot|s_n))+\pi_{\theta_n}(\cdot|s_n)\pi_{\theta_n}(\cdot|s_n)^T\right), \nonumber 
\end{align}
which implies
\begin{align}
    \langle\frac{\partial^2 R(\theta)}{\partial \theta_n(s_n,\cdot)^2} \nu_n(s_n,\cdot), \nu_{n}(s_{n},\cdot)\rangle \leq& \left| \langle\text{diag}(\pi_{\theta_n}(\cdot|s_n)) \nu_n(s_n,\cdot), \nu_{n}(s_{n},\cdot)\rangle\right| + \|\langle\pi_{\theta_n}(\cdot|s_n), \nu_n(s_n, \cdot) \rangle\|^2 \nonumber\\
    \leq& 2\|\nu_n(s_n, \cdot)\|^2. \nonumber
\end{align}
Finally we have $$\langle\frac{\partial^2 R(\theta)}{\partial \theta^2} \nu, \nu)\rangle \leq 2\|\nu\|^2,$$ 
which completes the proof.
\end{proof}

\subsection{Proving Theorem \ref{thm: inexact local}}
Recall $\eta \leq 1/\beta'$ and $D(t) = L^*(\rho) - L^{\pi_{\theta^t}}(\rho).$ Given $\theta^t$ and $\nabla \hat L^{\pi_{\theta^t}}(\rho),$ we have
\begin{align*}
    D(t+1) - D(t) =& L^{\pi_{\theta^t}}(\rho) - L^{\pi_{\theta^{t+1}}}(\rho) \\
    =& L^{\pi_{\theta^t}}(\rho) - L^{\pi_{\theta^{t+1}}}(\rho) +  \langle \nabla L^{\pi_{\theta^t}}(\rho), \theta^{t+1}-\theta^t\rangle -  \langle \nabla L^{\pi_{\theta^t}}(\rho), \theta^{t+1}-\theta^t\rangle \\
    \leq& \frac{\beta'}{2} ||\theta^{t+1}-\theta^t||^2 - \langle \nabla L^{\pi_{\theta^t}}(\rho), \theta^{t+1}-\theta^t\rangle \\
    =&  \frac{\beta' \eta^2}{2} || \nabla \hat L^{\pi_{\theta^t}}(\rho) ||^2 - \eta \langle \nabla L^{\pi_{\theta^t}}(\rho), \nabla \hat L^{\pi_{\theta^t}}(\rho)\rangle \\
    =& \frac{1}{2\beta'} || \nabla \hat L^{\pi_{\theta^t}}(\rho) - \nabla L^{\pi_{\theta^t}}(\rho)||^2 - \frac{1}{2\beta} || \nabla L^{\pi_{\theta^t}}(\rho) ||^2.
\end{align*}
Since $\epsilon_t = \mathbb E\left[|| \nabla \hat L^{\pi_{\theta^t}}(\rho) - \nabla L^{\pi_{\theta^t}}(\rho)||^2\right],$ we have
\begin{align*}
    \mathbb E\left[D(t+1) - D(t)\right] 
    \leq& \frac{1}{2\beta'}\mathbb E \left[ || \nabla \hat L^{\pi_{\theta^t}}(\rho) - \nabla L^{\pi_{\theta^t}}(\rho)||^2 \right] - \frac{1}{2\beta'} \mathbb E \left[|| \nabla L^{\pi_{\theta^t}}(\rho) ||^2\right],
\end{align*}
which implies
\begin{align*}
\frac{1}{2\beta'} \mathbb E \left[|| \nabla L^{\pi_{\theta^t}}(\rho) ||^2\right] 
    \leq& \frac{1}{2\beta'}\mathbb E \left[ || \nabla \hat L^{\pi_{\theta^t}}(\rho) - \nabla L^{\pi_{\theta^t}}(\rho)||^2 \right] + \mathbb E\left[D(t) - D(t+1)\right]\\
    =& \frac{\epsilon_t}{2\beta'} + \mathbb E\left[D(t) - D(t+1)\right].
\end{align*}
Take summation $t=1,2,\cdots, T$ and we have 
$$\sum_{t=1}^T \mathbb E \left[|| \nabla L^{\pi_{\theta^t}}(\rho) ||^2\right] \leq \sum_{t=1}^T \epsilon_t + 2\beta' \left(L^*(\rho) - L^{\pi^1}(\rho)\right).$$

\section{Proof of Theorem \ref{thm: main TD}}\label{app:TDMAPG}
Recall the estimated gradient and the true gradient to be
\begin{align*}
\nabla_{\theta_n} \hat{V}^{\pi_\theta}(\rho) =&
 \sum_{h=1}^H \gamma^h  \frac{1}{N} \sum_{k\in \mathcal N(n)} \hat \delta^{\pi_\theta,h}_k(s_{\mathcal N(k)}(h), a_{\mathcal N(k)}(h)) \nabla_{\theta_n}\log \pi_{\theta_n}(a_n(h)|s_n(h)),\\
\nabla_{\theta_n} V^{\pi_\theta}(\rho) =&
\sum_{h=1}^\infty \gamma^h \mathbb E\left[ \frac{1}{N} \sum_{k\in \mathcal N(n)}  \delta^{\pi_\theta,h}_k(s_{\mathcal N(k)}(h), a_{\mathcal N(k)}(h)) \nabla_{\theta_n}\log \pi_{\theta_n}(a_n(h)|s_n(h))\right].
\end{align*}
Further define $\nabla_{\theta_n} \tilde{V}^{\pi_\theta}(\rho)$ and $\nabla_{\theta_n} \tilde{V}^{\pi_\theta}(\rho)$ to decompose policy gradient as follows:
\begin{align*}
\nabla_{\theta_n} \tilde{V}^{\pi_\theta}(\rho) =&
 \sum_{h=1}^H \gamma^h  \frac{1}{N} \sum_{k\in \mathcal N(n)}  \delta^{\pi_\theta,h}_k(s_{\mathcal N(k)}(h), a_{\mathcal N(k)}(h)) \nabla_{\theta_n}\log \pi_{\theta_n}(a_n(h)|s_n(h)),\\
\nabla_{\theta_n} \bar V^{\pi_\theta}(\rho) =& \sum_{h=1}^H \gamma^h \mathbb E\left[ \frac{1}{N} \sum_{k\in \mathcal N(n)}  \delta^{\pi_\theta,h}_k(s_{\mathcal N(k)}(h), a_{\mathcal N(k)}(h)) \nabla_{\theta_n}\log \pi_{\theta_n}(a_n(h)|s_n(h))\right].
\end{align*}
The error is decomposed to be 
\begin{align*}
&\nabla_{\theta_n} \hat{L}^{\pi_\theta}(\rho) - \nabla_{\theta_n} {L}^{\pi_\theta}(\rho) \\
=& \underbrace{\nabla_{\theta_n} \hat{L}^{\pi_\theta}(\rho) - \nabla_{\theta_n} \tilde{L}^{\pi_\theta}(\rho)}_{e_{1,n}} + \underbrace{\nabla_{\theta_n} \tilde{L}^{\pi_\theta}(\rho) - \nabla_{\theta_n} \bar{L}^{\pi_\theta}(\rho)}_{e_{2,n}} 
+\underbrace{\nabla_{\theta_n} \bar{L}^{\pi_\theta}(\rho) - \nabla_{\theta_n} {L}^{\pi_\theta}(\rho)}_{e_{3,n}},
\end{align*}
where the error $e_{1,n}$ is related to the estimated TD-error and its true TD-error under policy $\pi_\theta;$ the error $e_{2,n}$ is related to the sample-based TD-error estimation and its expected one; the error $e_{3,n}$ is related to the truncated error of estimated state occupancy measure.

To prove Theorem \ref{thm: main TD}, we establish the following lemma that relates $\sum_{t=1}^T \mathbb E[ ||\nabla L^{\pi_{\theta^t}}(\rho)||^2 ]$ to the errors of $e_1, e_2,$ and $e_3.$
\begin{lemma}\label{lem: main TD local}
Under Algorithm \ref{alg: inexact PG}, we have 
\begin{align*}
\sum_{t=1}^T \mathbb E[ ||\nabla L^{\pi_{\theta^t}}(\rho)||^2 ] 
\leq& 3 \eta \sum_{t=1}^T \mathbb E\left[||e_1(t)||^2+||e_2(t)||^2+||e_3(t)||^2\right]\\
& + \left(\frac{2}{(1-\gamma)^2} + 2\lambda N\right)\sum_{t=1}^T\mathbb E\left[||e_1(t)|| +||e_3(t)||\right] \\
& + \left(\beta' \eta - 2\right) \sum_{t=1}^T \mathbb E\left[\langle \nabla L^{\pi_{\theta^t}}(\rho), e_2(t)\rangle \right] \\
& + \frac{2}{\eta} \left(L^*(\rho) - L^{\pi^1}(\rho)\right).
\end{align*}
\end{lemma}
\begin{proof}
Define $D(t) = L^*(\rho) - L^{\pi_{\theta^t}}(\rho).$ Therefore, we have 
\begin{align*}
D(t+1) - D(t) 
\leq& \frac{\beta' \eta^2}{2} || \nabla \hat L^{\pi_{\theta^t}}(\rho) ||^2 - \eta \langle \nabla L^{\pi_{\theta^t}}(\rho), \nabla \hat L^{\pi_{\theta^t}}(\rho)\rangle \\
=& \frac{\beta' \eta^2}{2} || \nabla L^{\pi_{\theta^t}}(\rho) + e_1(t) + e_2(t) + e_3(t)||^2 \\
&- \eta \langle \nabla L^{\pi_{\theta^t}}(\rho), \nabla L^{\pi_{\theta^t}}(\rho)+ e_1(t) + e_2(t) + e_3(t) \rangle\\
\leq& \left(\frac{\beta' \eta^2}{2} - \eta\right) ||\nabla L^{\pi_{\theta^t}}(\rho)||^2 + \frac{\beta' \eta^2}{2} ||e_1(t)+e_2(t)+e_3(t)||^2\\
& + \left(\frac{\beta' \eta^2}{2} - \eta\right) \langle \nabla L^{\pi_{\theta^t}}(\rho), e_1(t) + e_2(t) + e_3(t)\rangle 
\end{align*}
Recall the values of $\eta$ and $\beta',$ where $\beta' \eta \leq 1$, we have 
\begin{align*}
D(t+1) - D(t)
\leq& -\frac{\eta}{2} ||\nabla L^{\pi_{\theta^t}}(\rho)||^2 + \frac{3\eta}{2}\left(||e_1(t)||^2+||e_2(t)||^2+||e_3(t)||^2\right)\\
& + \eta||\nabla L^{\pi_{\theta^t}}(\rho)|| \left(||e_1(t)|| + ||e_3(t)||\right)\\ 
& + \left(\frac{\beta' \eta^2}{2} - \eta\right) \langle \nabla L^{\pi_{\theta^t}}(\rho), e_2(t)\rangle,
\end{align*}
which implies that 
\begin{align}
\sum_{t=1}^T \mathbb E[ ||\nabla L^{\pi_{\theta^t}}(\rho)||^2 ] 
\leq& 3 \eta \sum_{t=1}^T \mathbb E\left[||e_1(t)||^2+||e_2(t)||^2+||e_3(t)||^2\right] \nonumber\\
& + 2\sum_{t=1}^T\mathbb E\left[\|\nabla L^{\pi_{\theta^t}}(\rho)\|\left(||e_1(t)|| +||e_3(t)||\right)\right]  \nonumber\\
& + \left(\beta' \eta - 2\right) \sum_{t=1}^T \mathbb E\left[\langle \nabla L^{\pi_{\theta^t}}(\rho), e_2(t)\rangle \right] \nonumber\\
& + \frac{2}{\eta} \left(L^*(\rho) - L^{\pi^1}(\rho)\right). \label{eq:gradint bounds}
\end{align}
Next, we provide the upper bound on $\|\nabla L^{\pi_{\theta^t}}(\rho)\|.$ Note that 
\begin{align}
\delta^{\pi_{\theta},h}_n(s_{\mathcal N(n)}(h), a_{\mathcal N(n)}(h)) =& r^h_n(s_{\mathcal N(n)}(h), a_{\mathcal N(n)}(h)) + \gamma \hat V^{\pi_{\theta},H}_{n}(s_{\mathcal N(n)}(h+1)) - \hat V^{\pi_{\theta},H}_{n}(s_{\mathcal N(n)}(h))\nonumber\\
\leq& 1 + \frac{\gamma}{1-\gamma} = \frac{1}{1-\gamma}, \nonumber
\end{align}
and
\begin{align}
\nabla_{\theta_n(s'_n,a'_n)}\log \pi_{\theta_n}(a_n(h)|s_n(h)) = \mathbb I_{s_n(h)=s'_n}\left(\mathbb I_{a_n(h)=a'_n} - \pi_{\theta_n}(a_n'|s_n(h)) \right).\nonumber
\end{align}
We immediately have $$\|\nabla L^{\pi_{\theta^t}}(\rho)\| \leq \frac{2N}{(1-\gamma)^2} + 2\lambda N.$$
Finally, we complete the proof by substituting the upper bound into \eqref{eq:gradint bounds}.
\end{proof}
In the following subsections, we provide the upper bounds on the error terms in Lemma \ref{lem: main TD local}. We define $\{f-g\}(x):=f(x)-g(x)$ and occasionally abbreviate the index $t$ in $\pi_{\theta^t}$ for a simple notation. 
\subsection{Bounds on $\|e_1(t)\|^2,\|e_2(t)\|^2,$ and $\|e_3(t)\|^2$}
\begin{lemma}
$$||e_1(t)||^2+||e_2(t)||^2+||e_3(t)||^2 \leq \frac{12N}{(1-\gamma)^4}.$$
\end{lemma}
\begin{proof}
As in the proof of Lemma \ref{lem: main TD local}, it could be verified that $\|\nabla_{\theta_n}{V}^{\pi_\theta}(\rho)\|, \|\nabla_{\theta_n}\tilde{V}^{\pi_\theta}(\rho)\|,$
$\|\nabla_{\theta_n}\bar{V}^{\pi_\theta}(\rho)\|$ are bounded by the term
$$\frac{2}{(1-\gamma)^2}.$$ 
Therefore, all these error terms of $\|e_1\|^2,\|e_2\|^2,$ and $\|e_3\|^2$ are also bounded by $$\frac{4N}{(1-\gamma)^4},$$
which completes the proof.
\end{proof}
\subsection{Bound on $\langle \nabla L^{\pi_{\theta^t}}(\rho), e_2(t)\rangle$}
\begin{lemma}
$$\mathbb E\left[\left|\sum_{t=1}^T \langle \nabla L^{\pi_{\theta^t}}(\rho), e_2(t)\rangle\right|\right] \leq  \left(\frac{4N}{(1-\gamma)^4} + \frac{2\lambda N}{(1-\gamma)^2}\right) \sqrt{T\log T}.$$\label{lem: e3 cross term}
\end{lemma}
\begin{proof}
The sequence of $\{\langle \nabla L^{\pi_{\theta^t}}(\rho), e_2(t)\rangle\}_t$ is a martingale difference sequence with respect to $\mathcal F_{t-1}$ because it satisfies 
\begin{itemize}
\item $\left|\langle \nabla L^{\pi_{\theta^t}}(\rho), e_2(t) \rangle\right| \leq  ||\nabla L^{\pi_{\theta^t}}(\rho)|| ||e_2(t)|| \leq \frac{4N}{(1-\gamma)^4} + \frac{2\lambda N}{(1-\gamma)^2}.$
\item $\mathbb E[\langle \nabla L^{\pi_{\theta^t}}(\rho), e_2(t)\rangle | \mathcal F_{t-1}] = 0.$
\end{itemize}
Therefore, we use Azuma–Hoeffding inequality for the sequence such that
\begin{align}
    \mathcal P\left(\left|\sum_{t=1}^T \langle \nabla L^{\pi_{\theta^t}}(\rho), e_2(t)\rangle\right| \geq \left(\frac{4N}{(1-\gamma)^4} + \frac{2\lambda N}{(1-\gamma)^2}\right) \sqrt{2T\log T} \right) \leq \frac{1}{T}, \nonumber
\end{align}
which completes the proof.
\end{proof}
\subsection{Bound on $\|e_{3,n}(t)\|$}
\begin{lemma}
$$\mathbb E\left[||e_{3,n}(t)||\right] \leq \frac{2\gamma^H}{1-\gamma}, \forall 1\leq n \leq N, \forall 1\leq t \leq T.$$\label{lem: e2}
\end{lemma}
\begin{proof}
We have  
\begin{align*}
\|e_{3,n}(t)\| =& \left\|\sum_{h=H}^\infty \gamma^h \mathbb E\left[ \frac{1}{N} \sum_{k\in \mathcal N(n)}  \delta^{\pi_\theta}_k(s_{\mathcal N(k)}(h), a_{\mathcal N(k)}(h)) \nabla_{\theta_n}\log \pi_{\theta_n}(a_n(h)|s_n(h))\right]\right\| \nonumber \\
\leq&  \frac{4|\mathcal N(n)|}{(1-\gamma)N} \left\|\sum_{h=H}^\infty \gamma^h \mathbb E\left[\nabla_{\theta_n}\log \pi_{\theta_n}(a_n(h)|s_n(h))\right]\right\| \nonumber \\
\leq&  \frac{4\gamma^H}{1-\gamma}. \nonumber 
\end{align*}
\end{proof}

\subsection{Bounds on $\|e_{1,n}\|$}
\begin{lemma}
$$\mathbb E\left[||e_{1,n}(t)||\right] \leq  \frac{4\max_{k\in \mathcal N(n)}\|\hat V^{\pi_{\theta^t}}_k - V^{\pi_{\theta^t}}_k\|}{1-\gamma}, \forall 1\leq n \leq N, \forall 1 \leq t \leq T.$$\label{lem: e2}
\end{lemma}
\begin{proof}
Recall the error
\begin{align}
e_{1,n} =& \nabla_{\theta_n} \hat{L}^{\pi_\theta}(\rho) - \nabla_{\theta_n} \tilde{L}^{\pi_\theta}(\rho) \nonumber \\
=& \sum^{H}_{h=1} \gamma^h \frac{1}{N} \sum_{k \in \mathcal N(n)}\{\hat \delta^{\pi_\theta}_k - \delta^{\pi_\theta}_k\}(s_{\mathcal N(k)}(h), a_{\mathcal N(k)}(h)) \nabla_{\theta_n}\log \pi_\theta(a_n|s_n) \nonumber 
\end{align}
which implies that 
\begin{align}
\|e_{1,n}\| \leq & \sum^{H}_{h=1} \gamma^h \frac{1}{N} \sum_{k \in \mathcal N(n)}\left| \{\hat \delta^{\pi_\theta}_k - \delta^{\pi_\theta}_k\}(s_{\mathcal N(k)}(h), a_{\mathcal N(k)}(h))\right| \| \nabla_{\theta_n}\log \pi_\theta(a_n|s_n)\| \nonumber \\
\leq& \sum^{H}_{h=1} \gamma^h \frac{1}{N} \sum_{k \in \mathcal N(n)} \|\hat \delta^{\pi_\theta}_k - \delta^{\pi_\theta}_k\| \| \nabla_{\theta_n}\log \pi_\theta(a_n|s_n)\|\nonumber\\
\leq& 2\sum^{H}_{h=1} \gamma^h \max_{k\in \mathcal N(n)}\|\hat \delta^{\pi_\theta}_k - \delta^{\pi_\theta}_k\| \nonumber\\
\leq& \frac{4\max_{k\in \mathcal N(n)}\|\hat V^{\pi_\theta}_k - V^{\pi_\theta}_k\|}{1-\gamma}, \nonumber
\end{align} 
where the third inequality holds because $\| \nabla_{\theta_n}\log \pi_\theta(a_n|s_n)\| \leq 2$ and the last inequality holds because $|\mathcal N(n)| \leq N;$ and the last inequality holds because of the definition of TD-error $\delta.$ 
\end{proof}

Since $e_1(t)$ is directly related to the estimation error of value function, we follow \cite{SriYin_19} to establish a finite-time analysis of $\hat V^{\pi_{\theta^t},h}-V^{\pi_{\theta^t},h}.$
We write down a linear representation of value function updating as in \cite{SriYin_19}. Define $z_{n}^h = (s_{\mathcal N(n)}(h),a_{\mathcal N(n)}(h))$ and a vector $u(z_{n}^h)$ with the $z_{n}^h$th entry being one and zeros otherwise. We abbreviate $t$ and $n$ for the simple notation because we evaluate a fixed policy $\pi_{\theta^t}$ for any $t$ and the following derivation holds for any $n.$ We represent $\hat V^{\pi_{\theta},h} (z_{n}^h) = \langle \Theta^h, u(z_{n}^h) \rangle$ with parameter $\Theta^h \in \mathbb R^{|\mathcal S_n|\times |\mathcal A_n|}.$ 
Recall the value function updates in Algorithm \ref{alg: TD PG} as follows 
\begin{align*}
&\hat V^{\pi_{\theta},h+1}_{n} (s_{\mathcal N(n)}(h), a_{\mathcal N(n)}(h)) \\
= & (1-\alpha)\hat  V^{\pi_{\theta},h}_{n}(s_{\mathcal N(n)}(h)) + \alpha (r^h_n(s_{\mathcal N(n)}(h),a_{\mathcal N(n)}(h))+ \gamma \hat V^{\pi_{\theta},h}_{n}(s_{\mathcal N(n)}(h+1)), \nonumber 
\end{align*}
which can be wrote down in a linear form
\begin{align*}
\Theta^{h+1} = \Theta^h + \alpha \left( -u(z^h) (u^T(z^h) - \gamma u^T(z^{h+1})) \Theta^h + r(z^h) u(z^h)\right). 
\end{align*}
Define $A(z_h,z_{h+1}) = -u(z_h) (u^T(z_h) - \gamma u^T(z_{h+1})$ and $b(z_h) = r(z_h) u(z_h),$ we have 
\begin{align}
\Theta^{h+1} = \Theta^h + \alpha \left( A(z^h,z^{h+1}) \Theta^h + b(z^h) \right). \label{eq:Q-update} 
\end{align}
Define a matrix to be $\Phi$ with the column being $u(z^h)$ and $\Lambda$ to be a diagonal matrix with $\Lambda(z^h,z^h) = \pi(z^h),$ where $\pi(z^h)$ is the stationary distribution of Markov chain $Z^h,$ $\Gamma$ to be the transition kernel matrix and $r$ to be a vector with $r(z^h)$ being the $z^h$th entry. Further define $\tilde{A} = - \Phi \Lambda  (\Phi^T - \gamma \Gamma \Phi^T)$ and $\tilde{b} = \Phi \Lambda r.$ The corresponding (shift) ODE of \eqref{eq:Q-update} is 
$$\dot \theta = \tilde{A} \theta + \tilde{b},$$ where its equilibrium is $\theta^*$ and it is also the solution to Bellman equation, i.e., $$V^{\pi_{\theta},h} (z^h) = \langle \theta^*, u(z^h) \rangle.$$ Then we study $\|\Theta^h - \theta^*\|$ by leveraging the drift analysis as in \cite{SriYin_19}. 
The Lyapunov equation of ODE is  $-I = \tilde{A}^T P + P \tilde{A}$ with a positive symmetric $P,$ where $\xi_{\max}$ and $\xi_{\min}$ are the largest and smallest eigenvalues of $P.$

Based on Assumption \ref{assump: fast mixing}, let $\tau$ be the mixing time such that total variance distance of $\mathcal P(z^h | z^0)$ and $\pi(z^h)$ is less than $\frac{\xi_{\max}}{0.9}\frac{\log H}{H},$ that is 
$$\|\mathcal P(z^h | z^0) - \pi(z^h)\|_{TV} \leq \frac{\xi_{\max}}{0.9}\frac{\log H}{H}.$$ 
Based on these definitions, we state Theorem $7$ in \cite{SriYin_19} as follows.
\begin{theorem}[Theorem $7$ in \cite{SriYin_19}]
For any $h>\tau$ and $\epsilon$ such that $256\epsilon \tau + \epsilon \xi_{\max} \leq 0.05,$ we have the following finite-time
bound:
\begin{align}
    \mathbb E\left[\|\Theta^h - \theta^*\|^2\right] \leq \frac{\xi_{\max}}{\xi_{\min}} \left(1-\frac{0.9 \epsilon}{\xi_{\max}}\right)^{h-\tau} (1.5\|\Theta^0\| + 0.5)^2 + \frac{980\xi_{\max}^2}{\xi_{\min}}\epsilon\tau. \nonumber
\end{align}\label{thm: sriyin}
\end{theorem}
By setting $\epsilon=\frac{\xi_{\max}}{0.9}\frac{\log H}{H-\tau},$ we can invoke Theorem $7$ in \cite{SriYin_19} and have the following lemma.
\begin{lemma}
For any $H > \tau$ such that $\frac{256\xi_{\max}\tau\log H}{H-\tau} + \frac{\xi_{\max}^2\log H}{H-\tau} \leq 0.045,$ we have the estimation error of value functions  
\begin{align}
\mathbb E[\|\hat V^{\pi_{\theta},H} - V^{\pi_{\theta}}\|^2] \leq \frac{\xi_{\max}}{4\xi_{\min}} \frac{1}{H} + \frac{1090\xi_{\max}^3 }{\xi_{\min}} \frac{\tau \log H}{H-\tau}.  \nonumber 
\end{align}
\end{lemma}
\begin{proof}
Let $\epsilon=\frac{\xi_{\max}}{0.9}\frac{\log H}{H-\tau}$ and $\Theta_0 = 0$ in Theorem \ref{thm: sriyin}, we have 
\begin{align}
    \mathbb E\left[\|\Theta^H - \theta^*\|^2\right] \leq& \frac{\xi_{\max}}{4\xi_{\min}} \left(1-\frac{\log H}{H-\tau}\right)^{H-\tau} + \frac{1090\xi_{\max}^3}{\xi_{\min}}\frac{\tau \log H}{H-\tau}\nonumber\\
    \leq& \frac{\xi_{\max}}{4\xi_{\min}}\frac{1}{H}+ \frac{1090\xi_{\max}^3}{\xi_{\min}}\frac{\tau \log H}{H-\tau},\nonumber
\end{align}
which completes the proof because of the linear representation of the value function $\langle \Theta^h, u(z^h) \rangle.$ 
\end{proof}

\subsection{Proving Theorem \ref{thm: main TD}}
Let $H=T+\tau$ and $\eta = 1/\sqrt{T}.$ We consider a small $\lambda$ and a large $T$ such that $\lambda N \log A_{\max} \leq 1$ and $\log (1/\gamma) \geq \log T/T.$ Combine all lemmas together, we have 
\begin{align*}
&\sum_{t=1}^T \mathbb E[ ||\nabla L^{\pi_{\theta^t}}(\rho)||^2 ]\\ 
\leq& \frac{36 N\sqrt{T}}{(1-\gamma)^4} +  \left(\frac{8N}{(1-\gamma)^2} + 8\lambda N \right)\sum_{n=1}^N\frac{T\gamma^H}{(1-\gamma)^2} \\
& + \left(\frac{4N}{(1-\gamma)^4} + \frac{2\lambda N}{(1-\gamma)^2}\right) \sqrt{T\log T} + 2\left(L^*(\rho) - L^{\pi^1}(\rho)\right)\sqrt{T} \\
& + \left(\frac{8N}{(1-\gamma)^3} + \frac{8\lambda N}{1-\gamma}\right) \sqrt{\left({\frac{\xi_{\max}}{4\xi_{\min}}+\frac{1090\xi_{\max}^3 }{\xi_{\min}}}\right)\tau T \log (\tau+T)}\\
\leq& \frac{74}{(1-\gamma)^4} \sqrt{\left({\frac{\xi_{\max}}{4\xi_{\min}}+\frac{1090\xi_{\max}^3 }{\xi_{\min}}}\right)\tau T \log (\tau+T)}
\end{align*}
the last inequality holds because
\begin{align}
    \gamma^H \leq 1/\sqrt{T} ~~\text{and}~~ L^*(\rho) - L^{\pi^1}(\rho) \leq \frac{1}{1-\gamma} + \lambda N\log A_{\max}\nonumber,
\end{align}
which completes the proof of Theorem \ref{thm: main TD} by 
\begin{align*}
\min_{1 \leq t \leq T}\mathbb E[ ||\nabla L^{\pi_{\theta^t}}(\rho)||^2 ] \leq& \frac{1}{T}\sum_{t=1}^T \mathbb E[ ||\nabla L^{\pi_{\theta^t}}(\rho)||^2 ]\\ 
\leq& \frac{74}{(1-\gamma)^4} \sqrt{\left({\frac{\xi_{\max}}{4\xi_{\min}}+\frac{1090\xi_{\max}^3 }{\xi_{\min}}}\right)\frac{\tau  \log (\tau+T)}{T}}.
\end{align*}

\end{document}